\documentclass{article}
\usepackage[letterpaper, margin=1in]{geometry}
\usepackage{amsmath,amsfonts,amsthm}
\usepackage[backend=biber]{biblatex}
\usepackage{bm}
\usepackage{float}
\usepackage{mathtools}
\usepackage{graphicx} 
\usepackage[font=small,labelfont=bf]{caption}
\usepackage{enumitem}
\usepackage{booktabs}
\usepackage{authblk}

\newtheorem{prop}[subsection]{Proposition}
\newtheorem{lemma}[subsection]{Lemma}
\newtheorem{theorem}[subsection]{Theorem}
\newtheorem{corollary}[subsection]{Corollary}
\newtheorem{conjecture}[subsection]{Conjecture}
\newtheorem{remark}[subsection]{Remark}

\theoremstyle{definition}
\newtheorem{definition}[subsection]{Definition}

\newcommand{\C}{\mathbb{C}}
\newcommand{\R}{\mathbb{R}}

\title{Clifford and quadratic composite operators with applications to non-Hermitian physics}

\author{Jos\'{e} J. Garcia$^1$ \\ josejgarcia@unm.edu}
\affil{$^1$Department of Mathematics and Statistics, University of New Mexico, Albuquerque, NM, USA}
\date{September 2025}

\begin{document}

\maketitle

\begin{abstract}
    A variety of physical phenomena, such as amplification, absorption, and radiation, can be effectively described using non-Hermitian operators.
    However, the introduction of non-uniform non-Hermiticity can lead to the formation of exceptional points in a system's spectrum, where two or more eigenvalues become degenerate and their associated eigenvectors coalesce causing the underlying operator or matrix to become defective.
    Here, we explore extensions of the Clifford and quadratic $\epsilon$-pseudospectrum, previously defined for Hermitian operators, to accommodate non-Hermitian operators and matrices, including the possibility that the underlying operators may possess exceptional points in their spectra. In particular, we provide a framework for finding approximate joint eigenvectors of a $d$-tuple of Hermitian operators $\bm{A}$ and non-Hermitian operators $\bm{B}$, and show that their Clifford and quadratic $\epsilon$-pseudospectra are still well-defined despite any non-normality.
    We prove that the non-Hermitian quadratic gap is local with respect to the probe location when there are perturbations to one or more of the underlying operators.
    Altogether, this framework enables the exploration of non-Hermitian physical systems' $\epsilon$-pseudospectra, including but not limited to photonic systems where gain, loss, and radiation are prominent physical phenomena.
\end{abstract}

\section{Introduction}

Localized states in physical systems have a variety of important technological applications.
In multidimensional systems, states that are localized to a single dimension are useful for enabling directed transport of energy or information.
Moreover, zero-dimensional cavity states have uses in storing energy and information.
In systems described by Hermitian Hamiltonians, finding such localized states is generally accomplished by identifying (or designing) one or more states to be isolated in the system's spectrum, whose corresponding local density of states (LDOS) is spatially localized.
Indeed, finding classes of systems that exhibit such boundary-localized states is one of the primary motivations for studying material topology (See \cite{bansil_lin_das_RevModPhys.88.021004},\cite{hasan_kane_2010_RevModPhys.82.3045},\cite{ozawa_et_all_RevModPhys.91.015006},\cite{xiao_chang_niu_RevModPhys.82.1959}), where bulk-boundary correspondence can guarantee such localized states.

Open systems have the property that the corresponding Hamiltonian is non-Hermitian, stemming from either a loss of energy or decoherence due to the environment.
As such, for an open system, the Hamiltonian $H$ can have complex eigenvalues whose imaginary parts correspond to the loss (or gain) of energy or information in each associated eigenvector.
A canonical class of open systems are photonic systems, which can exhibit amplification, absorption, and radiation.
Here, without including microscopic descriptions of the sources of amplification or absorption, or considering the infinite space into which photons can radiate, these phenomena are instead generally considered by introducing non-Hermitian terms into the photonic system's Hamiltonian. Nevertheless, finding low-loss localized states in photonics, and open systems more broadly, is a perennial goal, as these states exhibit Purcell enhancement of any latent nonlinear effects \cite{purcell_spontaneous_1946}. 

Previously, the search for system states that are approximately localized in energy and position, i.e., that are approximate joint eigenvectors of both a system's Hamiltonian and position operators, has been studied using the system's quadratic pseudospectrum \cite{cerjan_loring_vides_2023}.
In such an approach,
the idea is to choose a $(d+1)$-tuple $\bm{\lambda} = (x_1,...,x_d, E)= (\bm{x},E)$ consisting of position and energy and extracting information from the quadratic composite operator
\begin{equation}\label{eq:Hermitian_quadratic_composite_operator}
    Q_{(\bm{x},E)}(\bm{X},H) = \sum_{i=1}^d (X_i-x_i)^2 + (H-E)^2.
\end{equation}
The smallest eigenvalue corresponds to a unit state that is an approximate joint eigenvector of each of the matrices $X_i$ and $H$.
In particular, the `probe site' $\bm{\lambda}$ in $\R^{d+1}$, provides the corresponding approximate eigenvalues of the matrices.
This mathematical method gives rise to the quadratic $\epsilon$-pseudospectrum and is the main tool in \cite{cerjan_loring_vides_2023}.
Similar approaches based on the system's Clifford $\epsilon$-pseudospectrum have also been used to both identify localized states in realistic photonic systems and develop a class of invariants for material topology \cite{cerjan_loring_pho_2022,kahlil_cerjan_loring,cerjan_prl_2024}.
By itself, the inclusion of non-Hermiticity does not preclude the possibility of finding localized states generally, nor boundary localized states protected by material topology see \cite{Esaki_Sato_Hasebe_Kohmoto_PR_2011},  \cite{kawabata_shiozaki_ueda_sato_PhysRevX.9.041015_2019}, \cite{Kawabata_etal_nature_2019}, and \cite{Shen_Zhen_Fu_PRL_2018}.
However, when working with non-Hermitian systems, we can fail to find an approximately localized state when directly applying the method outlined above with the $\epsilon$-pseudospectrum. 
The problem is that the quadratic composite operator in Equation \ref{eq:Hermitian_quadratic_composite_operator} can fail to be normal for non-Hermitian Hamiltonians $H$, and consequently could have degenerate eigenspaces.
Thus, the degeneracy of the eigenspaces and eigenvalues of Equation \ref{eq:Hermitian_quadratic_composite_operator} could imply the quadratic gap is not small enough.
Thus we could fail to find approximately localized states, since the error bounds guaranteeing a good approximately localized eigenvector depend on the magnitude of the quadratic gap function.

Pseudospectral methods are being used for non-Hermitian systems, see \cite{kiorpelidis2024scaling}. Here, we propose pseudospectral methods for predicting approximately localized states and resonances in non-Hermitian systems in both position and energy.
We modify the quadratic composite operator in \cite{cerjan_loring_vides_2023} to a generalized version of the following,
\begin{equation}\label{eq:non_Hermitian_quadratic}
    Q_{(\bm{x},E)}\left(\bm{X},H\right) = \sum_{i}^d (X_i-x_i)^2 + (H-E)^\dagger (H-E)
\end{equation}
while allowing $E$ to be complex valued since $H$ being non-Hermitian gives us the possibility of complex eigenvalues.
This updates definitions for a probe site as well as the $\epsilon$-pseudospectrum.
We take this idea and expand it to allow for multiple non-Hermitian operators in the construction, in particular non-normal operators ($AA^\dagger \not= A^\dagger A$).
We prove Lipschitz continuity of the corresponding gap functions.
We also prove that under perturbations to a physical system's Hamiltonian, the quadratic gap function allows us to find approximately localized states with small loss based on the probe site.
Along the way, we investigate the non-Hermitian spectral localizer and argue why the radial gap can be especially helpful in non-Hermitian systems with exceptional points.
We also make a connection between the different gap functions and highlight the theory with numerical simulations.
The theory and numerical simulations allow us to consider a system with an exceptional point and how the radial gap function helps spot the exceptional point.
We also use the theory for numerical simulations on a line gapped system.
Overall, one of our goals is to use the non-Hermitian pseudospectral methods we develop to design systems where we have some level of control over these long-lived lossy states.
In particular, we anticipate that a modified quadratic $\epsilon$-pseudospectrum will have applications in designing photonic structures for enhancing light-matter interactions.

The remainder of the paper is outlined as follows:
In Section \ref{sec:approx_jes_Hermitian},
we introduce the idea of joint eigenvectors between two matrices or observables and review known results for when all matrices involved are Hermitian. We finalize by highlighting some ideas from \cite{cerjan_loring_vides_2023}.
In Section \ref{sec:nh_clifford_composite_operator}, we consider a specific construction of the non-Hermitian spectral localizer in \cite{cerjan_koekebier_Schulz_Baldes}, which we call the non-Hermitian spectral localizer $L_{(\bm{\lambda},\nu)} \left( \bm{A}, B \right)$, and note some salient properties.
We also briefly review the non-Hermitian Clifford linear gap and introduce a new gap called the non-Hermitian Clifford radial gap; the idea being that the Clifford linear gap will be more useful for line gapped physical systems,
while the clifford radial gap could be more useful for point gapped systems.
Additionally, we prove some results about the gaps and finalize with a conjecture.
In Section \ref{sec:nh_quadratic_composite}, we motivate the construction of the quadratic composite operator $Q_{(\bm{\lambda},\nu)}\left(\bm{A},\bm{B}\right)$ in Equation \ref{eq:non_Hermitian_quadratic} by computing 
\begin{equation*}
L_{(\bm{\lambda},\nu)} \left( \bm{A}, B \right)^\dagger L_{(\bm{\lambda},\nu)} \left( \bm{A}, B \right).
\end{equation*}
We define a quadratic gap $\mu_{(\bm{\lambda},\nu)}^\textup{Q}\left(\bm{A},B\right)$ and prove its equivalence to other numerical values.
We end the section by proving Lipschitz continuity of the quadratic gap as well as provide important corollaries about joint approximate eigenvectors.
In Section \ref{sec:rel_spec_and_quad_gap}, we provide error estimates between the non-Hermitian Clifford linear gap, non-Hermitian Clifford radial gap, and quadratic gap.
In Section $\ref{sec:local_nature_nh_quadratic}$,
we show that the quadratic gap is local, meaning that far away perturbations of $B$ relative to probe site $(\bm{\lambda},\bm{\nu})$ will have little effect on the corresponding non-Hermitian quadratic gap at that point. 
In Section \ref{sec:applications_to_point_gapped_systems}, we apply the non-Hermitian Clifford linear gap, radial gap, and non-Hermitian quadratic gap functions to a two level non-Hermitian system.
We argue why the non-Hermitian radial gap function should be used instead of the gap mentioned in \cite{cerjan_koekebier_Schulz_Baldes} in the case of point gapped systems.
In Section \ref{sec:HHS_Lossy_Modes},
we take a similar tight binding model as the one used in \cite{cerjan_koekebier_Schulz_Baldes}, and numerically compute the Clifford linear and quadratic gap functions to visualize were we might find approximately localized states in position and energy.
Finally, we compare the Clifford linear gap to the quadratic gap function which highlights how closely the functions agree and where.

\section{Pseudospectrum for joint approximate eigenvectors}\label{sec:approx_jes_Hermitian}

The use of pseudospectral methods is not a new one.
Over the years it has been used either implicitly explicity.
Some early uses of these methods with an eye towards physics is mentioned in \cite{Arnold1972},\cite{landau_1976}, and \cite{landau_1977}.
The ideas of pseudospectra became well known in the 90's after Trefethen and Embree's book \cite{trefethen_embree_2005_MR2155029}.
This was the culmination of previous papers written by Trefethen, see \cite{Trefethen1992},\cite{trefethen_1997_edsjsr.213303719970901}.
It is important to note that we only lightly touch on the history of the pseudospectrum and more can be read in \cite{trefethen_embree_2005_MR2155029}.
When we consider the eigenvalue problem. In other words it is equivalent to asking if $(A-\lambda)^{-1}$ exists given a square matrix $A$ and some $\lambda \in \C$.
However for numerical purposes it is better to ask if $\lVert(A-\lambda)^{-1}\rVert$ is large.
This is how Trefethen and Embree motivate the pseudospectra of the matrix $A$ in \cite{trefethen_embree_2005_MR2155029}.
In particular the proposed definition is that for some $\epsilon>0$ the pseudospectrum of $A$ is the set of complex numbers $\lambda$ such that $\lVert(A-\lambda)^{-1}\rVert>\epsilon^{-1}$ for some small $\epsilon>0$.
In their book \cite{trefethen_embree_2005_MR2155029}, Trefethen and Embree give alternative definitions that are all equivalent for matrices.
Here we include them to for completeness.
\begin{definition}\label{def:te_first}
Suppose $A$ in $M_n(\C)$ and $\epsilon>0$. The \emph{$\epsilon$-pseudospectrum} denoted $\text{Spec}_\epsilon(A)$ of $A$ is the set of $z$ in $\C$ that satisfy $\lVert (A-z)^{-1}\rVert>\epsilon^{-1}$.
\end{definition}
\begin{definition}\label{def:te_second}
Given $A$ in $M_n(\C)$ the set $\text{Spec}_\epsilon(A)$ is the set of all $z$ in $\C$ such that $z$ in $\text{Spec}(A+E)$ for $E$ in $M_n(\C)$ with $\lVert E\rVert < \epsilon$.
\end{definition}
\begin{definition}\label{def:te_third}
Suppose $A$ in $M_n(\C)$, $\text{Spec}_\epsilon(A)$ is the set of all $z$ in $\C$ such that $\lVert (A-z)\bm{v}\rVert<\epsilon$ for some non-zero $\bm{v}$ satisfying $\lVert\bm{v}\rVert=1$
\end{definition}
\begin{definition}\label{def:te_fourth}
    Suppose $A$ in $M_n(\C)$, $\text{Spec}_\epsilon(A)$ is the set of all $z$ in $\C$ such that $\sigma_{\min} (A-z)<\epsilon$.
\end{definition}
It is important to note that for definitions \ref{def:te_first},\ref{def:te_second},\ref{def:te_third} can use matrix norms other than the operator norm. However all four definitions (Definition \ref{def:te_first},\ref{def:te_second},\ref{def:te_third},\ref{def:te_fourth}) are only equivalent when using the operator norm on matrices.
Definition \ref{def:te_third} and \ref{def:te_fourth} in their book show up in an analogous way in the present paper.
Moving beyond the works of Trefethen, the topic of joint spectra has been further investigated by Jhanjee in their PhD Thesis, see \cite{JHANJEE2017}.
They investigate joint spectra for both matrices and operators in the commuting case.
In the present paper we emphasize what happens in the non-commuting case for the matrix case.

The main technique of pseudospectra in the multi matrix setting is to give equal treatment to each operator in a physical system.
Instead of computing eigenvectors for a Hamiltonian and then computing expectation values, pseudospectral techniques combine all of the incompatible observables/matrices in the system into a single composite operator.
The spectrum of said composite operator gives us information about states approximately localized at some position and energy when we consider a system's position and Hamiltonian operators, see
\cite{krejcirik_siegl_tater_viola_2015_10.1063/1.4934378},\cite{loring_schulz-baldes_2017_12504413820170101},\cite{loring_schulz-baldes_2020_edsgcl.62678504020200301},\cite{loring},\cite{trefethen_1997_edsjsr.213303719970901},\cite{trefethen_embree_2005_MR2155029},\cite{trefethen_reddy_driscoll_1993_6dce06aa-3a1d-396a-ad07-b01838717a11}.
When the matrices commute it can be shown that there is a unit state that is an eigenvector (with potentially different eigenvalue) for each observable.
This is a standard result in the theory of linear algebra.
However, it is also known that incompatible observables do not commute; in other words we are not guaranteed of the existence of joint eigenvectors of the corresponding observables.
The main idea however is that we are going to develop a framework in which we can find a unit eigenvector that is an approximate joint eigenvector of multiple matrices/observables through pseudospectral methods.
In order to motivate the study of approximate joint eigenvectors and as a result the study of $\epsilon$-pseudospectra; we begin with the study of joint eigenvectors from commuting matrices/observables.

It is known that a $d$ tuple of matrices $\bm{A}$ pairwise commute if and only if there there is a unitary that simultaneously transforms the $d$ matrices to upper traingular by a single unitary $U$ see Theorem 2.3.3 in \cite{horn_johnson}.
This gives us the existence of at least one simultaneous eigenvector.
Hence, if we have two matrices $A,B$ and the commutator is zero ($[A,B]=0$) then there is a state $\bm{\psi}$ such that $A\bm{\psi} = \lambda_A \bm{\psi}$ and $B\bm{\psi}= \lambda_B \bm{\psi}$.
In the situation where each matrix in $\bm{A}$ is normal ($A_iA_i^\dagger = A_i^\dagger A_i$) the theorem tell us that $A_i$ can be simultaneously diagonalized by the same unitary transformation. 

If $H$ is some Hamiltonian and $X$ is a position observable, we know that in general $[X,H]\not= 0$.
In the special case of an atomic limit, we have $[X,H]=0$.
However, most materials are not in the atomic limit.
We have $[X,H]\not= 0$ because $\left[\frac{d}{dx},x\right]\not=0$.
The physics comes into play when we note that a Hamiltonian is usually a differential operator.
Thus if we have a physical system the best we can hope for are states that are approximately localized in position and energy.
In other words if we have some number of position observables $\bm{X}$, and corresponding Hermitian Hamiltonian $H$, Proposition II.1 in \cite{cerjan_loring_vides_2023} answers the question of looking for $(\bm{x},E)$, and $\bm{\psi}$ so that
\begin{equation}\label{eq:approx_ev}
    X_i\bm{\psi} \approx x_i\bm{\psi} \mbox{ , }\hspace{0.1in}
    H\bm{\psi} \approx E\bm{\psi}.
\end{equation}

Indeed the approach that is taken is to utilize pseudospectral methods. Broadly speaking pseudospectra is a mathematical tool to determine approximate spectrum for non-commuting matrices.
To find approximate joint eigenvectors of matrices or observables, the idea is to consider various eigenvalue problems $(A_i-\lambda_i)\bm{\psi}$, with $A_i$ Hermitian.
Then we take a non trivial Clifford representation $\Gamma_i^\dagger = \Gamma_i$, $\Gamma_i^2 = I$, $\Gamma_i\Gamma_k = -\Gamma_k\Gamma_i$ for $i\not=k$ to construct the so called spectral localizer mentioned in Section I of \cite{cerjan_loring_vides_2023} as well as Section 1 of \cite{loring}.
Given $\bm{A}$ in $\Pi_{i=1}^{d+1}\mathcal{M}_n(\C)$ with each $A_i$ being Hermitian operators, and \emph{probe site} $\bm{\lambda}\in \R^{d+1}$,
\begin{equation*}
    L_{\bm{\lambda}}(\bm{A}) = \sum_{i=1}^{d+1} (A_i-\lambda_i)\otimes \Gamma_i
\end{equation*}
is the \emph{spectral localizer}.

By properties of the tensor product and the $A_i$ matrices, the spectral localizer is itself a Hermitian matrix.
From here a so called spectral localizer gap is defined and has alternative characterizations. It is defined explicitly in Section I of \cite{cerjan_loring_vides_2023}. However it first appears implicitly in Definition 1.1 of a so called Clifford $\epsilon$-pseudospectrum in \cite{loring}.

Note that for a $d=3$ dimensional system we could use the following as a Clifford representation
\begin{equation*}
\Gamma_1 = \sigma_x \otimes \sigma_x, \Gamma_2 = \sigma_y \otimes \sigma_x, \Gamma_3 = \sigma_z \otimes \sigma_x, \Gamma_4 = I_2 \otimes \sigma_z.
\end{equation*}
This tells us that
\begin{equation*}
    L_{(\kappa x,\kappa y,\kappa z,E)}^{\textup{3d}}\left(\kappa X,\kappa Y,\kappa Z,\kappa H\right) = \kappa(X-x)\otimes \Gamma_1 + \kappa(Y-y)\otimes \Gamma_2 + \kappa(Z-z)\otimes \Gamma_3 + (H-E)\otimes \Gamma_4.
\end{equation*}
We could have also used other matrices that satisfy the Clifford relations, such as 
\begin{equation*}
\Gamma_1 = \sigma_x \otimes \sigma_x, \Gamma_2 = \sigma_y \otimes \sigma_x, \Gamma_3 = \sigma_z \otimes \sigma_x, \Gamma_4 = I_2 \otimes \sigma_y
\end{equation*}
for example.
This is proven in \cite{cerjan_loring_2024_127892} by Cerjan and Loring.

We will call $\sigma_{\min}(A)$ the smallest singular value of a matrix $A$. We then say that $\textup{Spec}(A)$ is the set of eigenvalues of an operator $A$.
In particular we will take the convention that the spectrum of an arbitrary matrix $A$ will be enumerated so that the eigenvalues are ordered by their real parts.
We will also assume for the remainder of the paper unless otherwise stated that a matrix norm is assumed to be the spectral/operator norm.
In other words if $A$ is a matrix then the assumed norm for matrices is
\begin{equation*}
\left\lVert A \right\rVert = \max_{\left\lVert\bm{\psi}\right\lVert = 1} \lVert A\bm{\psi}\rVert.
\end{equation*}

The smallest singular value of the spectral localizer plays a major role in an attempt at finding states that are localized in its various observables/operators.
For this reason, in this context it is given a name.
Specifically given $\bm{A}$ in $\Pi_{i=1}^{d+1}\mathcal{M}_n(\C)$ with each $A_i$ being Hermitian operators, and \emph{probe site} $\bm{\lambda}$ in  $\R^{d+1}$, the \emph{spectral localizer gap} is defined as
\begin{equation*}\mu^{\text{C}}_{\bm{\lambda}}(\bm{A}) = \sigma_{\min}\left(L_{\bm{\lambda}}\left(\bm{A}\right)\right).
\end{equation*}

Given $\bm{A}$ in $\Pi_{i=1}^{d+1}\mathcal{M}_n(\C)$ with each $A_i$ being Hermitian operators the \emph{Clifford $\epsilon$-pseudospectrum} is defined as
\begin{equation*}\Lambda_\epsilon^{\text{C}}(\bm{A}) = \left\{\bm{\lambda}\in \R^{d+1} \mid \mu_{\bm{\lambda}}^{\text{C}}\left(\bm{A}\right)\leq\epsilon\right\}
\end{equation*}

With all this in place, one of the main results of \cite{loring} is finding a joint approximate state with corresponding probe site each corresponding to an element in the $\epsilon$-pseudospectrum.
Lemma 1.2 and 1.3 in \cite{loring} give formal error estimates, on the quantities $\left\lVert A_i\bm{\psi} - \lambda_k\bm{\psi}\right\rVert$.
We prove a generalized version of these two lemmas in Section \ref{sec:nh_clifford_composite_operator}.
The main idea being that a small gap or the pseudospectrum provides a recipe to get our hands on a unit state $\bm{\psi}$ that is approximately localized in position and space.

Now suppose that the operators in question are in $\mathcal{M}_n(\C)$, The dimension of the spectral localizer is $n2^{\lfloor\frac{d+1}{2}\rfloor}\times n2^{\lfloor\frac{d+1}{2}\rfloor}$. Then for every two more matrices you have in your system the dimension of the spectral localizer becomes $n2^{\lfloor\frac{d+1}{2}\rfloor+1}\times n2^{\lfloor\frac{d+1}{2}\rfloor+1}$. For two and three observables the spectral localizer is in $\mathcal{M}_{2n}(\C)$, for 4 and 5 observables the spectral localizer is in $\mathcal{M}_{4n}(\C)$, etc.
We can reduce the computational requirements by considering the fact that for hermitian matrices $A$, that $A\bm{\psi}=0$ iff $A^2\bm{\psi}=0$.
Indeed if we square the spectral localizer we get 
\begin{equation*}
\left(L_{\bm{\lambda}}(\bm{A})\right)^2 = \left(\sum_{i=1}^{d+1}(A_i-\lambda_i)^2\otimes I\right) + \left(\sum_{i<k}[A_i,A_k]\otimes \Gamma_i\Gamma_k\right).
\end{equation*}
When the commutator terms are small in operator norm, the predominant part of the squared non-Hermitian spectral localizer is the first sum.
This is a block diagonal matrix
\begin{equation}\label{eq:repeat_blocks}
    \left(\sum_{i=1}^{d+1}(A_i-\lambda_i)^2\right)\otimes I.
\end{equation}
The eigenvalues of Equation (\ref{eq:repeat_blocks}) are determined by the repeated block matrix in Equation (\ref{eq:herm_quadratic}).
\begin{equation}\label{eq:herm_quadratic}
    \sum_{i=1}^{d+1}(A_i-\lambda_i)^2
\end{equation} 
Equation (\ref{eq:herm_quadratic}) is defined to be the \emph{quadratic composite operator} in Section 1 of \cite{cerjan_loring_vides_2023}.
Given $\bm{A}$ in $\Pi_{i=1}^{d+1}\mathcal{M}_n(\C)$ where each of the $A_i$ are Hermitian operators and \emph{probe site} $\bm{\lambda}$ in $\R^{d+1}$, the \emph{quadratic gap} is defined to be
\begin{equation*}
        \mu_{\bm{\lambda}}^{\text{Q}}(\bm{A}) = \sqrt{\sigma_{\min}\left(\text{Q}_{\bm{\lambda}}(\bm{A})\right)}.
\end{equation*}
The square root is a consequence of proving this quantity is related to the smallest singular value of a related matrix $M_{\bm{\lambda}}(\bm{A})$ which is the topic of Proposition II.1 in \cite{cerjan_loring_vides_2023}.
Note that by construction, the quadratic composite operator has real eigenvalues and is positive semidefinite. 
Informaly a key idea with this gap function is that
\begin{equation*}
    \mu_{\bm{\lambda}}^{\text{C}}(\bm{A}) = \sigma_{\min}\left(L_{\bm{\lambda}}\left(\bm{A}\right)\right) \approx \sigma_{\min}\left(\text{Q}_{\bm{\lambda}}\left(\bm{A}\right)\right)\leq \mu_{\bm{\lambda}}^{\text{Q}}\left(\bm{A}\right)
\end{equation*}
when $\mu_{\bm{\lambda}}^Q(\bm{A})\leq 1$.
This also provides a corresponding state $\bm{\psi}$ that is an approximate eigenvector of the constituent observables/matrices.
For a comparison between the spectral localizer gap and quadratic gap function, see Proposition II.4 in \cite{cerjan_loring_vides_2023}.
The main idea is that the Clifford and quadratic gaps are close if the commutators $[A_i,A_k]$ are small, in other words the matrices $A_i$ nearly pairwise commute.

\section{Non-Hermitian spectral localizer}\label{sec:nh_clifford_composite_operator}

As mentioned in the introduction, a property of photonic systems is that they exhibit non-Hermitian Hamiltonians.
This motivates the question:
What happens when $A_{d+1}=B$ is non-Hermitian in the construction of the spectral localizer and in the quadratic composite operator? Or in a physical system setting, what happens when we have a non-Hermitian Hamiltonian $H$ in the construction of the corresponding composite operators?
We take inspiration from the constructions and definitions in \cite{cerjan_koekebier_Schulz_Baldes} to motivate this Section.
In the process we also introduce a new gap function and explain some of its properties.

One option for a construction of an equivalent spectral localizer is to consider a matrix  decomposition of $B$ into the Hermitian part $B_{\textup{H}}$ and its skew-Hermitian part $B_{\textup{S}}$ shown in Equation (\ref{eq:Hermitian_skew-Hermitian})
\begin{equation}\label{eq:Hermitian_skew-Hermitian}
    B=B_{\textup{H}} + iB_{\textup{S}}, \quad B_{\textup{H}}=(B+B^\dagger)/2, \quad B_{\textup{S}} = (B-B^\dagger)/2i .
\end{equation}
Both $B_{\textup{H}}$ and $B_{\textup{S}}$ are Hermitian, and thus we can use these together with position observables in the non-Hermitian spectral localizer, and thus we would be able to find a unit state $\bm{\psi}$ that is approximately localized in position and energy given the theory above.
In particular we could construct
\begin{equation*}
    L_{(\bm{x},\textrm{Re}(E),\textrm{Im}(E))}(\bm{X},B_{\textup{H}},B_{\textup{S}}).
\end{equation*}
However this goes against the intuition that the non-Hermitian spectral localizer should have similar properties as the Hamiltonian, in this case hermiticity or lack thereof.
Also, nothing necessarily stops us from using $B$ in the construction of the spectral localizer in Ref \cite{loring, cerjan_loring_vides_2023} with complex $E$.
In other words $L_{(\bm{x},E)}(\bm{X},B)$.
This makes it so that the resulting operator is not only not Hermitian, it is also not normal.
In fact, we can use the non-normality to show that the (right) eigenvector associated to the smallest eigenvalue ($L_{(\bm{x},E)}(\bm{X},B)\bm{\psi} = \lambda\bm{\psi}$) will give access to an approximate (right) eigenvector of all the matrices $X_i$ and $B$. Likewise, a left eigenvector associated with the smallest associated left eigenvalue ($\bm{\phi}^{\dagger} L_{(\bm{x},E)}(\bm{X},B) = \bm{\phi}^\dagger \alpha$) will give an approximate left eigenvector of all $X_i$ and $H$.
However, we will be following the approach in Ref \cite{cerjan_koekebier_Schulz_Baldes} for the construction of the non-Hermitian spectral localizer.

In the construction of the spectral localizer in Section I of \cite{cerjan_loring_vides_2023}, an arbitrary Clifford representation can be chosen. However, in the following construction for the non-Hermitian spectral localizer we use a specific Clifford representation.
This is because we can choose $d$ even so there will be a Clifford element representation that anticommutes with all the other $d$ elements and is diagonal in a suitable representation.
We will use such a representation and replace the diagonal element by a unique choice that allows us to extract information about both left and right eigenvectors.
The construction of said Clifford representation is as follows.
Note that $\sigma_x,\sigma_y,\sigma_z$ are the Pauli spin matrices.
For $d=2$ we use $\sigma_x$ and $\sigma_y$ and for $d=1$ we just drop $\sigma_y$.
For $d=3$ we use $\Gamma_1=\sigma_x, \Gamma_2 = \sigma_y, \Gamma_3 = \sigma_z$.
The Construction for $d=5$ is $\Gamma_1 = \sigma_x \otimes \sigma_x$, $\Gamma_2 = \sigma_y \otimes \sigma_x$, $\Gamma_3 =\sigma_z \otimes \sigma_x$, $\Gamma_4 = I_2 \otimes \sigma_y$, $\Gamma_5 = I_2 \otimes \sigma_z$.
Then suppose that $\{\gamma_i\}_{i=1}^d$ represents a Clifford representation for $d$ odd, we construct a Clifford representation for $d+2$ by computing 
\begin{align}\label{eq:cliff_rep_construction}
    \Gamma_i &= \gamma_i\otimes \sigma_x \text{ for all } 1\leq i\leq d \nonumber \\
    \Gamma_{d+1} &= I_m\otimes \sigma_y \\\quad \Gamma_{d+2} &= I_m \otimes \sigma_z \nonumber
\end{align} with
\begin{equation}\label{eq:cliff_dim}
    m = 2^{\left\lfloor d/2 \right\rfloor}.
\end{equation}
If we want the Clifford representation for $d$ even, we just construct the $d+2$ Clifford representation and discard $I_m\otimes \sigma_y$.
We now proceed with the definition of the non-Hermitian spectral localizer as follows.
\begin{definition}
    Given $\bm{A}$ in $\Pi_{i=1}^d\mathcal{M}_n(\C)$ where each of the $A_i$ are Hermitian operators, a non-Hermitian operator $B$, and a \emph{probe site} $(\bm{\lambda},\nu)$ in $\R^d \oplus \C$ we define the \emph{non-Hermitian spectral localizer} to be
    \begin{equation*}
        L_{(\bm{\lambda},\nu)} \left( \bm{A}, B \right) = \sum_{i=1}^d (A_i-\lambda_i)\otimes \Gamma_i + (B-\nu)\otimes \begin{bmatrix} I_m & \\ & 0_m \end{bmatrix}+ (B-\nu)^\dagger \otimes \begin{bmatrix} 0_m & \\ & -I_m \end{bmatrix}
    \end{equation*}
    Where the specific matrices $\{\Gamma_i\}_{i=1}^{d+1}$ are constructed as in Equations \ref{eq:cliff_rep_construction}.
\end{definition}

Again, this is a special case of the spectral localizer in \cite{cerjan_koekebier_Schulz_Baldes}.
From here we can manually compute and determine that the non-Hermitian spectral localizer is not normal.
This can be seen below when we compute the product
\begin{equation}\label{eq:ldagger_l_computation}
\begin{aligned}
    \left(L_{(\bm{\lambda},\nu)} \left( \bm{A}, B \right)\right)^{\dagger}\left(L_{(\bm{\lambda},\nu)} \left( \bm{A}, B \right)\right) &= \left(\sum_{i=1}^d(A_i-\lambda_i)^2\otimes I\right) + \left(\sum_{i<j} \left[ A_i, A_j \right] \otimes \Gamma_i \Gamma_j\right) + F\\
    & \quad +(B-\nu)^{\dagger}(B-\nu) \otimes \begin{bmatrix} I_m & \\ & 0_m \end{bmatrix} + (B-\nu)(B-\nu)^{\dagger}\otimes \begin{bmatrix} 0_m & \\ & I_m\end{bmatrix},
\end{aligned}
\end{equation}
with
\begin{equation}\label{eq:f_term}
    F = \left(\sum_{i=1}^d(A_i-\lambda_i)(B-\nu)\otimes \Gamma_i \begin{bmatrix} I_m & \\ & 0 \end{bmatrix} \right) + h.c. +\left(\sum_{i=1}^d(A_i-\lambda_i)(B-\nu)^{\dagger}\otimes \Gamma_i \begin{bmatrix} 0 & \\ & -I_m \end{bmatrix} \right) + h.c.
\end{equation}
Then again we compute $\left(L_{(\bm{x},E)} \left( \bm{X}, B \right)\right)\left(L_{(\bm{x},E)} \left( \bm{X}, B \right)\right)^{\dagger}$ to show that 
\begin{equation*}
    \left[L_{(\bm{x},E)} \left( \bm{X}, B \right) ,\left(L_{(\bm{x},E)} \left( \bm{X}, B \right)\right)^{\dagger}\right]\not=0.
\end{equation*}
Note that $h.c.$ in Equation \ref{eq:f_term} represents the Hermitian conjugate of the matrix preceding the term.

Despite the non-Hermitian spectral localizer being non-normal, we do have a property that is just as useful if not better.
If $B$ is close to being Hermitian then so too is the non-Hermitian spectral localizer.
Formally we have that for all $\epsilon>0$, $\left\lVert B-B^\dagger \right\rVert < \epsilon$ implies that if we label $L = L_{(\bm{\lambda},\nu)} \left( \bm{A}, B \right)$, then
\begin{align*}
    \left\lVert L - L^\dagger \right\rVert &= \left\lVert\left((B-\nu)-(B-\nu)^\dagger\right)\otimes \begin{bmatrix} I_m & \\ & 0_m \end{bmatrix} + \left((B-\nu)^\dagger - (B-\nu)\right) \otimes \begin{bmatrix}0_m & \\ & \!\!\!\!\!-I_m \end{bmatrix} \right\rVert \\
    &= \left\lVert\left((B-\nu)-(B-\nu)^\dagger\right)\otimes I_m \right\rVert \\
    &= \left\lVert(B-\nu)-(B-\nu)^\dagger\right\rVert \\
    &<\epsilon + 2\left\lvert\textrm{Im}(\nu)\right\rvert.
\end{align*}

This also tells us that the closeness of the non-Hermitian spectral localizer to Hermitian is also dependent on how big of a value we probe in energy.
Further more we have that,
\begin{align*}
    \left\lVert L^\dagger L - LL^\dagger \right\rVert &\leq \left\lVert (L - L^\dagger) L \right\rVert + \left\lVert L (L - L^\dagger) \right\rVert \\
    &\leq \left\lVert L \right\rVert \left\lVert (L - L^\dagger) \right\rVert + \left\lVert L \right\rVert \left\lVert  (L - L^\dagger) \right\lVert \\
    &\leq 2\left\lVert L \right\rVert \left(\epsilon + 2\left\lvert\textrm{Im}(\nu)\right\rvert\right).
\end{align*}
which tells us that the non-Hermitian spectral localizer being close to hermitian also implies it is close to normal.
In a way this is not too surprising from the point of view that if a matrix $A$ is Hermitian, then it follows that $A$ is normal, indeed $A=A^\dagger$ implies $A^\dagger A = AA = AA^\dagger$.
To define the gap function for the non-Hermitian spectral localizer, we take the definition in Section 2 of \cite{cerjan_koekebier_Schulz_Baldes} that is used while studying line gapped systems.
In particular it is called the spectral localizer gap, but we will refer to it by a different name.
\begin{definition}
    Given $\bm{A}$ in $\Pi_{i=1}^d\mathcal{M}_n(\C)$ where each of the $A_i$ are Hermitian operators, a non-Hermitian matrix $B$, and a \emph{probe site} $(\bm{\lambda},\nu)$ in $\R^d \oplus \C$ we define the \emph{Clifford linear gap} to be
    \begin{equation*}\bar{\mu}_{(\bm{\lambda},\nu)}^{\text{C}}\left(\bm{A},B\right) = \min\big\vert \textrm{Re}\big(\textup{Spec}\left(L_{(\bm{\lambda},\nu)}\left(\bm{A},B\right)\right) \big) \big\vert.
    \end{equation*}
\end{definition}
Based on numerical results we conjecture that the Clifford linear gap function is continuous.
\begin{conjecture}
    The Clifford linear gap function is pointwise continuous.
\end{conjecture}
From a topological point of view, especially for line gapped systems, it seems to be instrumental to not look at $\min\left\lvert\textup{Spec}\left(L_{(\bm{\lambda},\nu)}\left(\bm{A},B\right)\right)\right\rvert$ but instead to flatten the spectrum and instead look at $\min\left\lvert\textrm{Re}(\textup{Spec}\left(L_{(\bm{\lambda},\nu)}\left(\bm{A},B\right)\right)\right\rvert$. 
When working with non-Hermitian systems, we have to take into consideration both point gapped and line gapped systems.
For this reason we propose the following definition for point gapped systems.
By numerical computations in Section \ref{sec:applications_to_point_gapped_systems} we will demonstrate that the smallest singular value of the non-Hermitian spectral localizer detects point gaps.
\begin{definition}
    Given $\bm{A}$ in $\Pi_{i=1}^d\mathcal{M}_n(\C)$ where each of the $A_i$ are Hermitian operators, a non-Hermitian matrix $B$, and a \emph{probe site} $(\bm{\lambda},\nu)$ in $\R^d \oplus \C$ we define the \emph{Clifford radial gap} to be
    \begin{equation*}\dot{\mu}_{(\bm{\lambda},\nu)}^{\text{C}}\left(\bm{A},B\right) = \sigma_{\min}\left(L_{\bm{(\lambda},\nu)}\left(\bm{A},B\right)\right).
    \end{equation*}
\end{definition}
A property of the Clifford radial gap is that it's Lipschitz continuous in a generalized taxi-cab metric.
This form of continuity implies the Clifford radial gap is differentiable almost everywhere.
This also means that when trying to generate a grid size for numerical simulations, it can be determined ahead of time based on the level of precision desired between data points.
\begin{prop}
    Given $\bm{A},\bm{C} \in \Pi_{i=1}^d\mathcal{M}_n(\C)$ where each of the $A_i$, $C_i$ are Hermitian operators, non-Hermitian matrices $B$, $D$, and probe sites $(\bm{\lambda},\nu), (\bm{\alpha},\beta)\in \R^d \oplus \C$ then
    \begin{equation*}
        \left|\dot{\mu}_{(\bm{\lambda},\nu)}^{\text{C}}\left(\bm{A},B\right) - \dot{\mu}_{{(\bm{\alpha},\beta)}}^{\text{C}}\left(\bm{C},D\right) \right| \leq 2\left[\left(\sum_{i=1}^d \left\lVert A_i-C_i\right\rVert\right) + \left\lVert B-D\right\rVert+ \left(\sum_{i=1}^d\left\lvert\lambda_i-\alpha_i\right\rvert \right) + \left\lvert\nu-\beta\right\rvert\right].
    \end{equation*}
\end{prop}

\begin{proof}
    First by triangle inequality we have that
    \begin{align*}
        \left\lVert L_{(\bm{\lambda},\nu)} \left( \bm{A}, B \right) - L_{(\bm{\alpha},\beta)} \left( \bm{C}, D \right) \right\rVert &\leq \sum_{i=1}^d \left\lVert\left((A_i-C_i)+(\alpha_i -\lambda_i)I_{2m}\right)\otimes \Gamma_i \right\rVert \\
        &\qquad+ \left\lVert\left((B-D)+(\beta-\nu)I_{m}\right)\otimes \begin{bmatrix} I_m & \\ & 0_m \end{bmatrix} \right\rVert \\
        &\qquad+ \left\lVert\left((B-D)^\dagger+(\beta-\nu)^*I_{m}\right)\otimes \begin{bmatrix} I_m & \\ & 0_m \end{bmatrix} \right\rVert \\
        &\leq \left(\sum_{i=1}^d \left\lVert(A_i-C_i)+(\alpha_i-\lambda_i)I_{2m}\right\rVert\right) + 2\left\lVert(B-D)+(\beta-\nu)I_{m}\right\rVert \\
        &\leq 2\left(\sum_{i=1}^d \left\lVert A_i-C_i\right\rVert\right) + 2\left\lVert B-D\right\rVert+ 2\left(\sum_{i=1}^d\left\lvert\lambda_i-\alpha_i\right\rvert \right) + 2\left\lvert\nu-\beta\right\rvert \\
    \end{align*}
    Lastly from \cite{bhatia} we know that
    \begin{equation*}
        \left\lvert\sigma_{\min} (A)- \sigma_{\min}(B)\right\rvert\leq \left\lVert A-B\right\rVert. 
    \end{equation*}
    Thus,
    \begin{equation*}
        \left|\dot{\mu}_{(\bm{\lambda},\nu)}^{\text{C}}\left(\bm{A},B\right) - \dot{\mu}_{{(\bm{\alpha},\beta)}}^{\text{C}}\left(\bm{C},D\right) \right| \leq \left\lVert L_{(\bm{\lambda},\nu)} \left( \bm{A}, B \right) - L_{(\bm{\alpha},\beta)} \left( \bm{C}, D \right) \right\rVert
    \end{equation*}
\end{proof}

The Clifford linear and radial gap functions is what we will used to define the corresponding $\epsilon$-pseudospectrum as follows.
\begin{definition}
    Given $\bm{A}$ in $\Pi_{i=1}^{d}\mathcal{M}_n(\C)$ with each $A_i$ being Hermitian, $B\in \mathcal{M}_n(\C)$ being non-Hermitian the \emph{Clifford linear $\epsilon$-pseudospectrum} is defined as
    \begin{equation*}
        \bar{\Lambda}_\epsilon^{\text{C}}(\bm{A},B) = \left\{(\bm{\lambda},\nu)\in \R^{d} \oplus \C \mid \bar{\mu}_{(\bm{\lambda},\nu)}^{\text{C}}\left(\bm{A},B\right)\leq\epsilon\right\}
    \end{equation*}
    and the \emph{Clifford radial $\epsilon$-pseudospectrum} defined as
    \begin{equation*}
        \dot{\Lambda}_\epsilon^{\text{C}}(\bm{A},B) = \left\{(\bm{\lambda},\nu)\in \R^{d} \oplus \C \mid \dot{\mu}_{(\bm{\lambda},\nu)}^{\text{C}}\left(\bm{A},B\right)\leq\epsilon\right\}.
    \end{equation*}
\end{definition}
Note that when $B$ is Hermitian and $\nu$ is real then we have
\begin{equation*}
    \Lambda_\epsilon^{\text{C}}(\bm{A},B) = \bar{\Lambda}_\epsilon^{\text{C}}(\bm{A},B) = \dot{\Lambda}_\epsilon^{\text{C}}(\bm{A},B).
\end{equation*}

The idea is that either the Clifford linear or radial $\epsilon$-pseudospectrum will provide a tuple $(\bm{\lambda},\nu)$ and a unit eigenvector $\bm{\psi}$ so that (\ref{eq:approx_ev}) holds.
As stated in Section \ref{sec:approx_jes_Hermitian} we will prove a similar formulation of Lemma 1.2 and 1.3 of \cite{loring}.
In other words the $\epsilon$-pseudospectrum gives rise to a unit state that is approximately localized in position and energy.

Note that Proposition \ref{prop:lm1.2_similar} tells us that we can find a state that is an approximate eigenvector of the matrices $A_i$ and $B$ or an approximate eigenvector for the matrices $A_i$ and $B^\dagger$.
In other words we either get an approximate right eigenvector for all $A_i$, and $B$ or an approximate left eigenvector of all $A_i$ and $B$.
When $B$ is non-normal, we know the left and right eigenvectors of $B$ are not conjugate to each other, however if B is normal then the left and right eigenvectors will coincide. In any case, it is not possible to find a joint approximate eigenvector for all the matrices $A_i$, $B$ and $B^\dagger$ without taking a loss in how good the approximate eigenvector is as a left or right eigenvector.

\begin{prop}\label{prop:lm1.2_similar}
    If $(\bm{\lambda},\nu)$ in $\dot{\Lambda}_{\epsilon_1}^{\text{C}}\left(\bm{A},B\right)$ for some $\epsilon_1 \geq 0$ and 
    \begin{equation*}
        \left(\sum_{i\not=k}\left\lVert [A_i,A_k]\right\rVert\right)+ \lVert F\rVert \leq \epsilon_2
    \end{equation*}
    for some $\epsilon_2\geq 0$ then there is a unit state $\bm{\psi}$ such that either
    \begin{equation*}
        \sqrt{\left(\sum_{i=1}^d \left\lVert A_i\bm{\psi} - \lambda_i \bm{\psi}\right\rVert^2\right) + \left\lVert B\bm{\psi}-\nu \bm{\psi}\right\rVert^2}\leq \sqrt{2m}\sqrt{\epsilon_1^2+ \epsilon_2},
    \end{equation*}
    or
    \begin{equation*}
        \sqrt{\left(\sum_{i=1}^d \left\lVert A_i\bm{\psi} - \lambda_i \bm{\psi}\right\rVert^2\right) +  \left\lVert B^\dagger\bm{\psi}-\nu^* \bm{\psi}\right\rVert^2}\leq \sqrt{2m}\sqrt{\epsilon_1^2+ \epsilon_2},
    \end{equation*}
    with $F$ as in Equation (\ref{eq:f_term}) and $m$ as in Equation (\ref{eq:cliff_dim}).
\end{prop}

\begin{proof}
    Suppose that we have selected $\Gamma_i$ is constructed as above so that $g= 2m$ with $m$ as in Equation \ref{eq:cliff_dim}.
    Without loss of generality we will consider probe site $(\bm{\lambda},\nu) = (\bm{0},0)$.
    Suppose then that $\epsilon = \dot{\mu}_{(\bm{0},0)}^{\text{C}}\left(\bm{A},B\right)$.
    Since $\epsilon$ is a singular value of $L_{(\bm{\lambda},\nu)} \left( \bm{A}, B \right)$, there is a corresponding right singular vector $\bm{\Psi}$.
    From here we have that $\bm{\Psi}$ is an eigenvector of $L_{(\bm{\lambda},\nu)} \left( \bm{A}, B \right)^\dagger L_{(\bm{\lambda},\nu)} \left( \bm{A}, B \right)$ with eigenvalue $\epsilon^2$.
    By equations (\ref{eq:ldagger_l_computation}),(\ref{eq:f_term}) and triangle inequality we have that
    \begin{align*}
        \left\lVert \left(\left(\sum_{i=1}^dA_i^2\otimes I_{2m}\right) + C\right)\bm{\Psi}\right\rVert &\leq \left\lVert L_{(\bm{\lambda},\nu)} \left( \bm{A}, B \right)^\dagger L_{(\bm{\lambda},\nu)} \left( \bm{A}, B \right) \bm{\Psi} \right\rVert \\
        &\quad + \left\lVert \left(\left(\sum_{i<j} \left[ A_i, A_j \right] \otimes \Gamma_i \Gamma_j\right) + F \right)\bm{\Psi} \right\rVert \\
        &\leq \epsilon^2+ \left(\sum_{i<k}\left\lVert [A_i,A_k]\right\rVert \right)+ \lVert F\rVert
    \end{align*}
    with
    \begin{equation*}
        C = B^{\dagger}B \otimes \begin{bmatrix} I_m & \\ & 0_m \end{bmatrix} + BB^{\dagger}\otimes \begin{bmatrix} 0_m & \\ & I_m\end{bmatrix}.
    \end{equation*}
    Now suppose that
    \begin{equation*}
        \bm{\Psi} = \begin{bmatrix} \bm{\Psi}_1 \\ \vdots \\ \bm{\Psi}_g \end{bmatrix}.
    \end{equation*}
    So that each $\bm{\Psi}_i$ is in $\C^{n}$ with $n$ being the dimension of the operators in $\bm{A}$.
    Then suppose that $k$ is the index maximizing $\lVert \bm{\Psi}_k \rVert$ so that $\lVert \bm{\Psi}_k \rVert \geq 1/g$.
    Now we will define $\bm{\psi} = \bm{\Psi}_k / \lVert \bm{\Psi}_k \rVert$.
    It follows that,
    \begin{equation*}
        \left\lVert \left[\left(\sum_{i=1}^dA_i^2\otimes I_{2m}\right) + C\right]\bm{\Psi}\right\rVert \geq \begin{cases}
            \left\lVert \left[\left(\sum_{i=1}^{d} A_i^2\right)+B^\dagger B\right] \bm{\Psi}_k \right\rVert \text{ if } 1\leq k\leq m \\
            \\
            \left\lVert \left[\left(\sum_{i=1}^{d} A_i^2\right)+BB^\dagger\right] \bm{\Psi}_k \right\rVert \text{ if } m< k \leq 2m.
        \end{cases}
    \end{equation*}
    implies
    \begin{equation*}
        g\left\lVert \left(\left(\sum_{i=1}^dA_i^2\otimes I_{2m}\right) + C\right)\bm{\Psi}\right\rVert \geq \begin{cases}
            \left\lVert \left(\left(\sum_{i=1}^{d} A_i^2\right)+B^\dagger B\right) \bm{\psi} \right\rVert \text{ if } 1\leq k\leq m\\
            \\
            \left\lVert \left(\left(\sum_{i=1}^{d} A_i^2\right)+BB^\dagger\right) \bm{\psi} \right\rVert \text{ if } m < k \leq 2m.
        \end{cases}
    \end{equation*}
    Next we have that if $1\leq k\leq m$
    \begin{equation*}
        \left(\sum_{i=1}^d\left\lVert A_i \bm{\psi} \right\rVert^2\right) + \left\lVert B \bm{\psi} \right\rVert^2 = \left\langle \left[\left(\sum_{i=1}^{d} A_i^2\right)+B^\dagger B\right] \bm{\psi} , \bm{\psi}\right\rangle \leq \left\lVert \left(\left(\sum_{i=1}^{d} A_i^2\right)+B^\dagger B\right) \bm{\psi} \right\rVert.
    \end{equation*}
    and if $m<k\leq 2m$
    \begin{equation*}
        \left(\sum_{i=1}^d\left\lVert A_i \bm{\psi} \right\rVert^2\right) + \left\lVert B^\dagger \bm{\psi} \right\rVert^2 = \left\langle \left[\left(\sum_{i=1}^{d} A_i^2\right)+BB^\dagger\right] \bm{\psi} , \bm{\psi}\right\rangle \leq \left\lVert \left(\left(\sum_{i=1}^{d} A_i^2\right)+BB^\dagger\right) \bm{\psi} \right\rVert.
    \end{equation*}
    This completes the proof.
\end{proof}

We will use the following inequality which is a direct consequence of triangle inequality, norms on tensor products as well as Equations (\ref{eq:ldagger_l_computation}) and (\ref{eq:f_term}); in particular, given a unit vector $\Psi$ we have that

\begin{align}
        \left\lVert L_{(\bm{\lambda},\nu)} \left( \bm{A}, B \right)^\dagger L_{(\bm{\lambda},\nu)} \left( \bm{A}, B \right) \bm{\Psi} \right\rVert &\leq \left\lVert \left(\left(\sum_{i=1}^d A_i^2\otimes I_{2m}\right) + C\right)\bm{\Psi}\right\rVert \nonumber\\
        &\quad + \left\lVert \left(\left(\sum_{i<j} \left[ A_i, A_j \right] \otimes \Gamma_i \Gamma_j\right) + F \right)\bm{\Psi} \right\rVert. \label{eq:ldagger_l_vector_ub}
\end{align}
    Lastly note that when the non-Hermitian spectral localizer is invertible,
\begin{align*}
        \left\lVert L_{(\bm{\lambda},\nu)} \left( \bm{A}, B \right)^{-1} \right\rVert^2 &= \left\lVert \left(L_{(\bm{\lambda},\nu)} \left( \bm{A}, B \right)^\dagger L_{(\bm{\lambda},\nu)} \left( \bm{A}, B \right)\right)^{-1} \right\rVert \\
        \left\lVert L_{(\bm{\lambda},\nu)} \left( \bm{A}, B \right)^{-1} \right\rVert^{-1} &= \sqrt{\left\lVert \left(L_{(\bm{\lambda},\nu)} \left( \bm{A}, B \right)^\dagger L_{(\bm{\lambda},\nu)} \left( \bm{A}, B \right)\right)^{-1} \right\rVert^{-1}} \\
        &= \sqrt{\sigma_{\min}\left(L_{(\bm{\lambda},\nu)} \left( \bm{A}, B \right)^\dagger L_{(\bm{\lambda},\nu)} \left( \bm{A}, B \right)\right)}
\end{align*}
This implies that 
\begin{align}
    \dot{\mu}_{(\bm{\lambda},\nu)}^{\text{C}}\left(\bm{A},B\right) &=\left\lVert L_{(\bm{\lambda},\nu)} \left( \bm{A}, B \right)^{-1} \right\rVert^{-1} \nonumber\\
    &= \sqrt{\min_{\left\lVert \mathbf{v} \right\rVert = 1} \left\lVert \left(L_{(\bm{\lambda},\nu)} \left( \bm{A}, B \right)^\dagger L_{(\bm{\lambda},\nu)} \left( \bm{A}, B \right) \right) \mathbf{v} \right\rVert} \nonumber\\
    &\leq \sqrt{\left\lVert L_{(\bm{\lambda},\nu)} \left( \bm{A}, B \right)^\dagger L_{(\bm{\lambda},\nu)} \left( \bm{A}, B \right) \bm{\Psi} \right\rVert}.\label{eq:radial_norm_ub}
\end{align}

Proposition \ref{lm:cliff_approx_2} states that if we have a unit state that is approximately a right left eigenvector of $A_i$ and $B$ then the corresponding tuple of eigenvalues lie in some $\epsilon$-pseudospectrum.

\begin{lemma}\label{lm:cliff_approx_2}
    Suppose that $\bm{A}$ is in $\Pi_{i=1}^{d+1}\mathcal{M}_n(\C)$ with each $A_i$ being Hermitian, and further suppose $B $ is in $\mathcal{M}_n(\C)$ and is non-Hermitian.
    If there is a unit state $\bm{\psi}$ such that
    \begin{align*}
        2\left(\sum_{i=1}^d \left\lVert A_i\bm{\psi}-\lambda_i\bm{\psi} \right\rVert^2 \right)& + \left\lVert B\bm{\psi}-\nu\bm{\psi} \right\rVert^2 + \left\lVert B^\dagger\bm{\psi}-\nu^*\bm{\psi} \right\rVert^2 \leq \epsilon_1 \\
        &\mbox{and} \\
        \left(\sum_{i\not=k}\left\lVert [A_i,A_k]\right\rVert \right)&+ \lVert F\rVert \leq \epsilon_2
    \end{align*}
    for all $i$, with $F$ as in Equation \ref{eq:f_term}; then $(\bm{\lambda},\nu)$ is in $\dot{\Lambda}_{\epsilon}(\bm{A},B)$ with $\epsilon = \sqrt{\frac{1}{\sqrt{2}}\epsilon_1 + \epsilon_2}$
     or $(\bm{\lambda},\nu)$ is in $\dot{\Lambda}_{\epsilon=0}(\bm{A},B)$.
\end{lemma}

\begin{proof}
    If the non-Hermitian spectral localizer is singular then we are done and $(\bm{\lambda},\nu)$ is in $\Lambda_{\epsilon=0}(\bm{A},B)$.
    So we will assume that the non-Hermitian spectral localizer is invertible.
    We will use linearity to reduce to the case where $(\bm{\lambda},\nu) = \bm{0}$.
    If $2\sum_{i=1}^d \left\lVert A_i^2\bm{\psi}\right\rVert + \left\lVert B^\dagger B \bm{\psi} \right\rVert + \left\lVert B B^\dagger \bm{\psi} \right\rVert\leq \epsilon_1$ for some unit $\bm{\psi}$ we will define
    \begin{equation*}
        \bm{\Psi} = \frac{1}{\sqrt{2m}}\phi^{\oplus 2m} = \frac{1}{\sqrt{2m}} \begin{bmatrix} \bm{\psi} \\ \vdots \\ \bm{\psi} \end{bmatrix}
    \end{equation*}
    with $m$ as in Equation \ref{eq:cliff_dim}.
    Then by Equation \ref{eq:ldagger_l_computation}, \begin{equation*}
        C = B^{\dagger}B \otimes \begin{bmatrix} I_m & \\ & 0_m \end{bmatrix} + BB^{\dagger}\otimes \begin{bmatrix} 0_m & \\ & I_m\end{bmatrix}
    \end{equation*}
    and $F$ as in Equation \ref{eq:f_term}, we have
    \begin{align*}
        \left\lVert \left(\left(\sum_{i=1}^d A_i^2\otimes I_{2m}\right) + C\right)\bm{\Psi}\right\rVert^2 &= \frac{1}{2} \left[\left\lVert\sum_{i=1}^d A_i^2\bm{\psi} + B^\dagger B \bm{\psi} \right\rVert^2 + \left\lVert\sum_{i=1}^d A_i^2\bm{\psi} + B B^\dagger \bm{\psi} \right\rVert^2\right]  \\
        &\leq \frac{1}{2} \left[\left\lVert\sum_{i=1}^d A_i^2\bm{\psi} + B^\dagger B \bm{\psi} \right\rVert + \left\lVert\sum_{i=1}^d A_i^2\bm{\psi} + B B^\dagger \bm{\psi} \right\rVert\right]^2  \\
        &\leq \frac{1}{2} \left[\sum_{i=1}^d \left\lVert A_i^2\bm{\psi}\right\rVert + \left\lVert B^\dagger B \bm{\psi} \right\rVert + \sum_{i=1}^d \left\lVert A_i^2\bm{\psi}\right\rVert + \left\lVert B B^\dagger \bm{\psi} \right\rVert\right]^2  \\
        &\leq \frac{1}{2}\epsilon_1^2
    \end{align*}
    combined with Equation \ref{eq:ldagger_l_vector_ub} that
    \begin{equation*}
        \left\lVert L_{(\bm{\lambda},\nu)} \left( \bm{A}, B \right)^\dagger L_{(\bm{\lambda},\nu)} \left( \bm{A}, B \right) \bm{\Psi} \right\rVert
        \leq \frac{1}{\sqrt{2}}\epsilon_1 + \left(\sum_{i\not=k}\left\lVert [A_i,A_k]\right\rVert \right)+ \lVert F\rVert.
    \end{equation*}

    Lastly Equation (\ref{eq:radial_norm_ub}) finalizes the proof.
\end{proof}

\begin{corollary}\label{cor:cliff_approx_rev}
    Given $\bm{A}$ in $\Pi_{i=1}^{d+1}\mathcal{M}_n(\C)$ with each $A_i$ being Hermitian, $B $ in $\mathcal{M}_n(\C)$ non-Hermitian.
    If there is a unit state $\bm{\psi}$ such that
    \begin{align*}
        \left(\sum_{i=1}^d \left\lVert A_i\bm{\psi}-\lambda_i\bm{\psi} \right\rVert^2 \right)& + \left\lVert B\bm{\psi}-\nu\bm{\psi} \right\rVert^2  \leq \epsilon_1 \\
        &\mbox{and} \\
        \left(\sum_{i\not=k}\left\lVert [A_i,A_k]\right\rVert \right)&+ \lVert F\rVert \leq \epsilon_2
    \end{align*}
    for all $i$, with $F$ as in Equation \ref{eq:f_term}; then $(\bm{\lambda},\nu)$ is in $\dot{\Lambda}_{\epsilon}(\bm{A},B)$ with $\epsilon = \sqrt{\epsilon_1 + \epsilon_2}$
     or $(\bm{\lambda},\nu)$ is in $\dot{\Lambda}_{\epsilon=0}(\bm{A},B)$.
\end{corollary}

\begin{proof}
    If the non-Hermitian spectral localizer is singular then we are done and $(\bm{\lambda},\nu)$ is in $\Lambda_{\epsilon=0}(\bm{A},B)$.
    So we will assume that the non-Hermitian spectral localizer is invertible.
    We will use linearity to reduce to the case where $(\bm{\lambda},\nu) = \bm{0}$.
    If $\sum_{i=1}^d \left\lVert A_i^2\bm{\psi}\right\rVert + \left\lVert B^\dagger B \bm{\psi} \right\rVert \leq \epsilon_1$ for some unit $\bm{\psi}$ we will define repeat the proof of Lemma (\ref{lm:cliff_approx_2}), but instead defining
    \begin{equation*}
        \bm{\Psi} = \frac{1}{\sqrt{m}}\phi^{\oplus m}\oplus \bm{0}^{m} = \frac{1}{\sqrt{m}} \begin{bmatrix} \bm{\psi} \\ \vdots \\ \bm{\psi} \\ \bm{0}  \\\vdots \\ \bm{0} \end{bmatrix}.
    \end{equation*}
    with $m$ as in Equation \ref{eq:cliff_dim}.
    Then by Equation \ref{eq:ldagger_l_computation}, and $F$ as in Equation \ref{eq:f_term}, we have
    \begin{align*}
        \left\lVert \left(\left(\sum_{i=1}^d A_i^2\otimes I_{2m}\right) + C\right)\bm{\Psi}\right\rVert^2 &= \left\lVert\sum_{i=1}^d A_i^2\bm{\psi} + B^\dagger B \bm{\psi} \right\rVert^2 \\
        &\leq \left[\sum_{i=1}^d \left\lVert A_i^2\bm{\psi}\right\rVert + \left\lVert B^\dagger B \bm{\psi} \right\rVert\right]^2  \\
        &\leq \epsilon_1^2.
    \end{align*}
    The rest of the proof uses the techniques of the proof in Lemma (\ref{lm:cliff_approx_2}).
\end{proof}

\begin{corollary}\label{cor:cliff_approx_lev}
    Given $\bm{A}$ in $\Pi_{i=1}^{d+1}\mathcal{M}_n(\C)$ with each $A_i$ being Hermitian, $B $ in $\mathcal{M}_n(\C)$ non-Hermitian.
    If there is a unit state $\bm{\psi}$ such that
    \begin{align*}
        \left(\sum_{i=1}^d \left\lVert A_i\bm{\psi}-\lambda_i\bm{\psi} \right\rVert^2 \right)& + \left\lVert B^\dagger\bm{\psi}-\nu^*\bm{\psi} \right\rVert^2  \leq \epsilon_1 \\
        &\mbox{and} \\
        \left(\sum_{i\not=k}\left\lVert [A_i,A_k]\right\rVert \right)&+ \lVert F\rVert \leq \epsilon_2
    \end{align*}
    for all $i$, with $F$ as in Equation \ref{eq:f_term}; then $(\bm{\lambda},\nu)$ is in $\Lambda_{\epsilon}(\bm{A},B)$ with $\epsilon = \sqrt{\epsilon_1 + \epsilon_2}$
     or $(\bm{\lambda},\nu)$ is in $\Lambda_{\epsilon=0}(\bm{A},B)$. Here $\Lambda_{\epsilon}(\bm{A},B)$ can be either the radial or linear $\epsilon$-pseudospectrum.
\end{corollary}

\begin{proof}
    If the non-Hermitian spectral localizer is singular then we are done and $(\bm{\lambda},\nu)$ is in $\Lambda_{\epsilon=0}(\bm{A},B)$.
    So we will assume that the non-Hermitian spectral localizer is invertible.
    We will use linearity to reduce to the case where $(\bm{\lambda},\nu) = \bm{0}$.
    If $\sum_{i=1}^d \left\lVert A_i^2\bm{\psi}\right\rVert + \left\lVert B B^\dagger \bm{\psi} \right\rVert \leq \epsilon_1$ for some unit $\bm{\psi}$ we will define repeat the proof of Lemma (\ref{lm:cliff_approx_2}), but instead defining
    \begin{equation*}
        \bm{\Psi} = \frac{1}{\sqrt{m}}\bm{0}^{m}\oplus \phi^{\oplus m} = \frac{1}{\sqrt{m}} \begin{bmatrix}  \bm{0}  \\\vdots \\ \bm{0} \\ \bm{\psi} \\ \vdots \\ \bm{\psi} \end{bmatrix}.
    \end{equation*}
    with $m$ as in Equation \ref{eq:cliff_dim}.
    Then by Equation \ref{eq:ldagger_l_computation}, and $F$ as in Equation \ref{eq:f_term}, we have
    \begin{align*}
        \left\lVert \left(\left(\sum_{i=1}^d A_i^2\otimes I_{2m}\right) + C\right)\bm{\Psi}\right\rVert^2 &= \left\lVert\sum_{i=1}^d A_i^2\bm{\psi} + B B ^\dagger \bm{\psi} \right\rVert^2 \\
        &\leq \left[\sum_{i=1}^d \left\lVert A_i^2\bm{\psi}\right\rVert + \left\lVert B B^\dagger \bm{\psi} \right\rVert\right]^2  \\
        &\leq \epsilon_1^2
    \end{align*}
    The rest of the proof uses the techniques of the proof in Lemma (\ref{lm:cliff_approx_2}).
\end{proof}

Note that in the situation where $BB^\dagger = B^\dagger B$, we can go through the proof of Lemma \ref{lm:cliff_approx_2} and see that we can drop the term $\left\lVert B^\dagger\bm{\psi}-\overline{\nu}\bm{\psi} \right\rVert$ in the assumptions. A similar result and proof for the Clifford linear $\epsilon$-pseudospectrum is not clear, we leave it as a conjecture.
\begin{conjecture}
    If $(\bm{\lambda},\nu)$ in $\bar{\Lambda}_{\epsilon}^{\text{C}}\left(\bm{A},B\right)$ then there is unit state $\bm{\psi}$ so that
    \begin{equation}
    X_i\bm{\psi} \approx x_i\bm{\psi} \mbox{ , }\hspace{0.1in}
    H\bm{\psi} \approx E\bm{\psi}.
\end{equation}
\end{conjecture}
The lack of a proof does not hinder us from proceeding with the generalization of the theory in \cite{cerjan_loring_vides_2023}.
The fact that for a normal matrix $A$, 
\begin{equation*}
    \sigma_{\min}(A) = \min|\textup{Spec}(A)|
\end{equation*}
makes establishing a relation between the spectral localizer gap and the quadratic gap a reasonable task with not too much work.
With non-normal matrices $A$ we can have that
\begin{equation*}
    \sigma_{\min}(A) < \min|\textup{Spec}(A)|.
\end{equation*}
In any case we are working with a non-normal matrix, and with luck the matrix can be nearly Hermitian or nearly normal.
However we do not work with luck and thus we need a way to characterize what we mean by nearly normal.
We will make use of definition and terminology presented in \cite{ruhe} in order to characterize deviation from normality.
In particular we have,
\begin{definition}
    The deviation from normality of a square matrix $A$ is defined as 
    \begin{equation*}
        \Delta_2(A) = \Delta(A) = \inf\{\vert\vert U\vert\vert : A = Q(D+U)Q^\dagger \mbox{ a Schur Decomposition}\}.
    \end{equation*}
\end{definition}

This is in some sense equivalent to considering the quantity $\vert\vert AA^\dagger - A^\dagger A\vert\vert$ and formal statements comparing the two quantities can be found in Theorem 2 of \cite{elsner_paardekooper}.
Also note that when the non-Hermitian spectral localizer does not deviate too much from normal then we get some decent spectral variation properties.

\section{Quadratic composite operator} \label{sec:nh_quadratic_composite}
In \cite{cerjan_loring_vides_2023} and \cite{loring}, it was shown that when it comes to finding unit states that are approximate joint eigenvectors to the observables/operators, both the spectral localizer's $\epsilon$-pseudospectrum and the quadratic's $\epsilon$-pseudospectrum provide good methods in the setting when all observables are Hermitian.
In Section \ref{sec:nh_clifford_composite_operator} we have a non-Hermitian matrix/Hamiltonian in the construction of the non-Hermitian spectral localizer in order to generalize the spectral localizer $\epsilon$-pseudospectrum theory.
The appeal of the spectral localizer $\epsilon$-pseudospectrum in the Hermitian case and non-Hermitian spectral localizer $\epsilon$-pseudospectrum in the non-Hermitian case is that it also provides information about a system's topology see \cite{cerjan_loring_pho_2022},\cite{cerjan_loring_2022},\cite{kahlil_cerjan_loring} for the Hermitian setting and \cite{cerjan_koekebier_Schulz_Baldes} for the non-Hermitian setting.
The quadratic $\epsilon$-pseudospectrum seems to be more helpful with finding localized states as described, with some added benefits.
The construction of the quadratic composite operator is inspired by Cerjan, Loring and Vides in \cite{cerjan_loring_vides_2023}, where it is defined for systems described by Hermitian Hamiltonians.
The quadratic composite operator first appears implicitly in \cite{loring}.
Because of the assumption taken in these prior studies that the constituent operators in the construction are Hermitian and the `probe site', i.e., the choice of $(\bm{x},E)$ is real, it is guaranteed that the corresponding $\epsilon$-pseudospectrum is real i.e. $\bm{\lambda}$ in $\R^{n+1}$.

One key idea is that for $A$ Hermitian, $A\bm{\Psi} = 0$ if and only if $A^2\bm{\Psi} = 0$.
Thus, it is reasonable to consider a similar procedure with the non-Hermitian spectral localizer in \cite{cerjan_koekebier_Schulz_Baldes}.
It's explicit and quick to prove that $A^{\dagger}A\bm{\Psi}=0$ if and only if $A\bm{\Psi}=0$.
We know that the singular values of the $L_{(\bm{\lambda},\nu)} \left( \bm{A}, B \right)$ are the square roots of the eigenvalues of 
\begin{equation*}
\left(L_{(\bm{\lambda},\nu)} \left( \bm{A}, B \right)\right)^{\dagger}\left(L_{(\bm{\lambda},\nu)} \left( \bm{A}, B \right)\right).
\end{equation*}
In particular this is useful since a small deviation from normality of the non-Hermitian spectral localizer implies
\begin{equation*}
    \dot{\mu}_{(\bm{\lambda},\nu)}^\text{C}(\bm{A},B) \approx \min\left|\textup{spec}\left(L_{(\bm{\lambda},\nu)} \left( \bm{A}, B \right)\right)\right|
\end{equation*}
and small deviation from being Hermitian of the non-Hermitian spectral localizer implies
\begin{equation*}
    \bar{\mu}_{(\bm{\lambda},\nu)}^\text{C}(\bm{A},B) \approx \min\left|\textup{spec}\left(L_{(\bm{\lambda},\nu)} \left( \bm{A}, B \right)\right)\right|
\end{equation*}
with formal estimates proven in Section \ref{sec:rel_spec_and_quad_gap}.

Thus we have that we can approximate the Clifford linear and radial gaps by the smallest singular value of the non-Hermitian spectral localizer.
By Equation \ref{eq:ldagger_l_computation} we have a way of defining a quadratic composite operator for the non-Hermitian spectral localizer similar to the Hermitian setting in \cite{cerjan_loring_vides_2023}.
We do this by ignoring $F$ and the commutator terms in the computation of Equation \ref{eq:ldagger_l_computation}.
Specifically we could use
\begin{equation*}
    \left(\sum_{i=1}^{d}(A_i-\lambda_i)^2 \otimes I_{2m} \right) + (B-\nu)^{\dagger}(B-\nu) \otimes \begin{bmatrix} I_m & \\ & 0_m \end{bmatrix} + (B-\nu)(B-\nu)^{\dagger} \otimes \begin{bmatrix} 0_m & \\ & I_m \end{bmatrix}.
\end{equation*}
Since the spectrum is determined by the diagonal block entries we could take the two distinct diagonal blocks.
Namely,
\begin{equation*}
    \sum_{i=1}^{d}(A_i-\lambda_i)^2 \otimes I_2 + (B-\nu)^{\dagger}(B-\nu) \otimes \begin{bmatrix} 1 & \\ & 0 \end{bmatrix} + (B-\nu)(B-\nu)^{\dagger}\otimes \begin{bmatrix} 0 & \\ & 1 \end{bmatrix}.
\end{equation*}

Instead of taking this as the definition of a quadratic composite operator, we will generalize this construction. We will allow various non-normal operators in the construction besides $B-\nu$. The proof that is to come about approximately localized eigenvectors will follow with this more general construction.

\begin{definition}
    Suppose that we have tuples of matrices $\bm{A}$ in $\Pi_{i=1}^{d_1} M_n(\C)$, $\bm{B}$ in $\Pi_{i=1}^{d_2} M_n(\C)$ with $A_i$ Hermitian, $B_j$ non-Hermitian, and \emph{probe site} $(\bm{\lambda},\bm{\nu})$ in $\R^{d_1}\oplus \C^{d_2}$ then we define the \emph{quadratic composite operator} as
    \begin{align*}
        Q_{(\bm{\lambda},\bm{\nu})}\left(\bm{A},\bm{B}\right) &= RQ_{(\bm{\lambda},\bm{\nu})}\left(\bm{A},\bm{B}\right)\otimes\begin{bmatrix} 1 & \\ & 0 \end{bmatrix} + LQ_{(\bm{\lambda},\bm{\nu})}\left(\bm{A},\bm{B}\right)\otimes\begin{bmatrix} 0 & \\ & 1 \end{bmatrix}.
    \end{align*}
    With the right and left quadratic composite operators defined as
    \begin{align*}
        RQ_{(\bm{\lambda},\bm{\nu})}\left(\bm{A},\bm{B}\right) &= \left(\sum_{i=1}^{d_1} (A_i-\lambda_i)^2 \right) + \left(\sum_{j=1}^{d_2} (B_j-\nu_j)^\dagger (B_j-\nu_j) \right) \\
        LQ_{(\bm{\lambda},\bm{\nu})}\left(\bm{A},\bm{B}\right) &= \left(\sum_{i=1}^{d_1} (A_i-\lambda_i)^2 \right) + \left(\sum_{j=1}^{d_2} (B_j-\nu_j) (B_j-\nu_j)^\dagger\right)
    \end{align*}
\end{definition}
Immediately from the definition we get the following Corollaries.
\begin{corollary} 
    Suppose that we have tuples of matrices $\bm{A} \in \Pi_{i=1}^{d_1}M_n(\C), \bm{B}\in \Pi_{i=1}^{d_2}M_n(\C)$ with $A_i$ Hermitian, $B_j$ non-Hermitian, and $(\bm{\lambda},\bm{\nu})$ in $\R^{d_1}\oplus \C^{d_2}$, then
    \begin{equation*}
        \sigma_{\min}\left(Q_{(\bm{\lambda},\bm{\nu})}\left(\bm{A},\bm{B}\right)\right) = \min\left\{\sigma_{\min}\left(RQ_{(\bm{\lambda},\bm{\nu})}\left(\bm{A},\bm{B}\right)\right),\sigma_{\min}\left(LQ_{(\bm{\lambda},\bm{\nu})}\left(\bm{A},\bm{B}\right)\right)\right\}
    \end{equation*}
\end{corollary}
Note that when all the $B_i$ are normal, the spectrum of $LQ_{(\bm{\lambda},\bm{\nu})}\left(\bm{A},\bm{B}\right)$ is equal to the spectrum of $RQ_{(\bm{\lambda},\bm{\nu})}\left(\bm{A},\bm{B}\right)$ which gives us the following corollary.
\begin{corollary}
    Suppose that we have tuples of matrices $\bm{A} \in \Pi_{i=1}^{d_1}M_n(\C), \bm{B}\in \Pi_{i=1}^{d_2}M_n(\C)$ with $A_i$ Hermitian, $B_j$ normal. If $(\bm{\lambda},\bm{\nu})$ is in $\R^{d_1}\oplus \C^{d_2}$, then
    \begin{equation*}
        \sigma_{\min}\left(Q_{(\bm{\lambda},\bm{\nu})}\left(\bm{A},\bm{B}\right)\right) = \sigma_{\min}\left(RQ_{(\bm{\lambda},\bm{\nu})}\left(\bm{A},\bm{B}\right)\right) = \sigma_{\min}\left(LQ_{(\bm{\lambda},\bm{\nu})}\left(\bm{A},\bm{B}\right)\right)
    \end{equation*}
\end{corollary}
Reflecting on the main point of Lemma \ref{prop:lm1.2_similar} we expect that the eigenvector corresponding with the smallest eigenvalue of the quadratic operator will pick up either an approximate right or left eigenvector of $A_i$ and $B_j$.
Using the definition of the quadratic composite operator above we have generalized the quadratic composite operator in \cite{cerjan_loring_vides_2023}. In fact the following proposition generalizes Proposition II.1. in \cite{cerjan_loring_vides_2023}. 

Recall that a unit state $\bm{\psi}$ in a Hilbert space determines a probability distribution with respect to a normal matrix $M$.
In particular for Hermitian matrices $M$ it is well established that we need only a complete basis for the eigenspace of $M$ to define a probablity distribution and as a result to define expectation and variance, see Appendix A of \cite{hirvensalo_2001}. While we do not need to use the probability distribution directly, we do need the corresponding expectation and square variance for both Hermitian and non-Hermitian matrices.
This allows us to talk loosely about approximate localization of a state in position, energy, momentum, frequency, etc.
The mathematical notation that we will use for \emph{expectation} and \emph{variance} are as follows,
\begin{align*}
    \textup{E}_{\bm{\psi}}[B] &= \langle B\bm{\psi}, \bm{\psi}\rangle \\
    \textup{V}^2_{\bm{\psi}}[B] &= \langle B^\dagger B\bm{\psi}, \bm{\psi} \rangle -\left\vert \langle B\bm{\psi}, \bm{\psi} \rangle \right\vert^2.
\end{align*}
Indeed when we have that $B=A$ is Hermitian then this reduces to the usual expectation and variance of a state $\bm{\psi}$ with respect to the appropriate observable/matrix.
With that in place we proceed to prove some results that link minimization of expectation and variance of matrices/observables to the quadratic composite operator.

\begin{prop} \label{prop:nh_rquad_co_equiv_quantities}
    Suppose that we have tuples of matrices $\bm{A} \in \Pi_{i=1}^{d_1}M_n(\C), \bm{B}\in \Pi_{i=1}^{d_2}M_n(\C)$ with $A_i$ Hermitian, $B_j$ non-Hermitian, and $(\bm{\lambda},\bm{\nu})$ in $\R^{d_1}\oplus \C^{d_2}$, the following are equal:
    \begin{enumerate}
        \item The quantity
        \begin{equation*}
            \min_{\vert\vert\bm{\psi}\vert\vert=1} \sqrt{\sum_{i=1}^{d_1} \vert\vert A_i\bm{\psi} - \lambda_i \bm{\psi} \vert\vert^2 + \sum_{j=1}^{d_2}\vert\vert B_j\bm{\psi} - \nu_{j} \bm{\psi} \vert\vert^2}.
        \end{equation*} \label{eq:unit_vect}

        \item The quantity
        \begin{equation*}
            \min_{\vert\vert\bm{\psi}\vert\vert=1} \sqrt{\sum_{i=1}^{d_1} \textup{V}^2_{\bm{\psi}}[A_i] + \left(\textup{E}_{\bm{\psi}}[A_i] -\lambda_i \right)^2 + \sum_{j=1}^{d_2}\textup{V}^2_{\bm{\psi}}[B_j] + \left\vert\textup{E}_{\bm{\psi}}[B_j] -\nu_j \right\vert^2},
        \end{equation*} 
        if $B_i$ is normal for all $1\leq i\leq d_2$.\label{eq:expectation_variance}

        \item The smallest singular value of
        \begin{equation*}
            RM_{(\bm{\lambda},\bm{\nu})} \left(\bm{A},\bm{B}\right) = \begin{bmatrix} A_1-\lambda_1 \\ \vdots \\ A_{d_1} - \lambda_{d_1} \\ B_1-\nu_{1} \\ \vdots \\ B_{d_2}-\nu_{d_2}\end{bmatrix} .
        \end{equation*}\label{eq:smin_M}

        \item The square root for the smallest eigenvalue of $RQ_{(\bm{\lambda},\bm{\nu})}\left(\bm{A},\bm{B}\right)$.\label{eq:eig_val_Q}
    \end{enumerate}
\end{prop}

\begin{proof}
To show (\ref{eq:unit_vect}) $\Leftrightarrow$ (\ref{eq:expectation_variance}) we first notice that $A$ is Hermitian and $\bm{\lambda} \in \R^{d_1}$, which implies that$(A-\lambda_i)^2 = (A-\lambda_i)^\dagger (A-\lambda_i)$. With this, it is enough to prove the statement for $RQ_{(\vec{\bm{0}},\bm{\nu})}(\bm{0},B)$.
\begin{align*}
    \vert\vert B_j\bm{\psi} - \nu_{j} \bm{\psi} \vert\vert^2 &= \langle B_j\bm{\psi} -\nu_j\bm{\psi},B_j\bm{\psi} -\nu_j \bm{\psi}\rangle \\
    &= \langle B_j\bm{\psi}, B_j\bm{\psi} \rangle - \overline{\nu_j}\langle B_j\bm{\psi}, \bm{\psi} \rangle - \nu_j\langle \bm{\psi}, B_j\bm{\psi} \rangle + |\nu_j|^2 \langle \bm{\psi}, \bm{\psi}\rangle \\
    &= \left(\langle B_j^\dagger B_j\bm{\psi}, \bm{\psi} \rangle -\left\vert \langle B_j\bm{\psi}, \bm{\psi} \rangle \right\vert^2\right) + \left(\left\vert \langle B_j\bm{\psi}, \bm{\psi} \rangle \right\vert^2 - \overline{\nu_j}\langle B_j\bm{\psi}, \bm{\psi} \rangle - \nu_j\langle \bm{\psi}, B_j\bm{\psi} \rangle + |\nu_j|^2\right) \\
    &= \left(\langle B_j^\dagger B_j\bm{\psi}, \bm{\psi} \rangle -\left\vert \langle B_j\bm{\psi}, \bm{\psi} \rangle \right\vert^2\right) + \left\vert \langle B_j\bm{\psi}, \bm{\psi} \rangle - \nu_j \right\vert^2 \\
    &= \textup{V}^2_{\bm{\psi}}[B_j] + \left\vert\textup{E}_{\bm{\psi}}[B_j] -\nu_j \right\vert^2.
\end{align*}

To show (\ref{eq:smin_M}) $\Leftrightarrow$ (\ref{eq:eig_val_Q}) we note that for any matrix $A$, the singular values of $A$ are the square roots of the eigenvalues of $A^{\dagger}A$. We combine this with the computation
\begin{equation*}
    \left(RM_{(\bm{\lambda},\bm{\nu})} \left(\bm{A},\bm{B}\right)\right)^{\dagger} \left(RM_{(\bm{\lambda},\bm{\nu})} \left(\bm{A},\bm{B}\right)\right) = RQ_{(\bm{\lambda},\bm{\nu})}\left(\bm{A},\bm{B}\right).
\end{equation*}
To show (\ref{eq:unit_vect}) $\Leftrightarrow$ (\ref{eq:smin_M}) we note by \cite{golub_vanloan}
\begin{align*}
    \sigma_{\min}\left(RM_{(\bm{\lambda},\bm{\nu})} \left(\bm{A},\bm{B}\right)\right) &= \min_{\vert\vert\bm{\psi}\vert\vert=1}\left\lVert RM_{(\bm{\lambda},\bm{\nu})} \left(\bm{A},\bm{B}\right)\bm{\psi}\right\rVert \\
    &= \min_{\vert\vert\bm{\psi}\vert\vert=1} \left\lVert \begin{bmatrix} A_1\bm{\psi}-\lambda_1\bm{\psi} \\ \vdots \\ A_{d_1}\bm{\psi} - \lambda_{d_1}\bm{\psi} \\ B_1\bm{\psi}-\nu_{1}\bm{\psi} \\ \vdots \\ B_{d_2}\bm{\psi}-\nu_{d_2}\bm{\psi}\end{bmatrix} \right\rVert\\
    &= \min_{\vert\vert\bm{\psi}\vert\vert=1} \sqrt{\sum_{i=1}^{d_1} \vert\vert A_i\bm{\psi} - \lambda_i \bm{\psi} \vert\vert^2 + \sum_{j=1}^{d_2}\vert\vert B_j\bm{\psi} - \nu_{j} \bm{\psi} \vert\vert^2}.
\end{align*}
\end{proof}

\begin{remark}
    A similar statement can be proven about the square root of the smallest eigenvalue of $LQ_{(\bm{\lambda},\bm{\nu})}\left(\bm{A},\bm{B}\right)$ by noting that $LQ_{(\bm{\lambda},\bm{\nu})}\left(\bm{A},\bm{B}\right) = RQ_{(\bm{\lambda},\bm{\nu}^{*})}\left(\bm{A},\bm{B}^\dagger\right)$.
\end{remark}

Thus if we want a state that is approximately localized in position and energy we have reason to study the smallest eigenvalue of the right quadratic composite operator.
If we want an approximate left eigenvector then we look at the smallest eigenvalue of the left quadratic composite operator.

\begin{definition}
    Suppose that $\bm{A}$ in $\Pi_{i=1}^{d_1}\mathcal{M}_n(\C)$, $\bm{B}$ in $\Pi_{i=1}^{d_2}\mathcal{M}_n(\C)$ with $A_i$ Hermitian, $B_j$ non-Hermitian and \emph{probe site} $(\bm{\lambda},\bm{\nu})$ in $\R^{d_1} \oplus \C^{d_2}$, then the \emph{right quadratic gap} is defined as
    \begin{equation*}
    \mu_{(\bm{\lambda},\bm{\nu})}^{\text{RQ}}\left(\bm{A},\bm{B}\right)=\sqrt{\sigma_{\min}\left(RQ_{(\bm{\lambda},\bm{\nu})}\left(\bm{A},\bm{B}\right)\right) }.
    \end{equation*}
\end{definition}
We can similarly define the left quadratic gap as being the square root of the smallest eigenvalue of the left quadratic composite operator and the quadratic gap function can be define as the square root of the smallest eigenvalue quadratic composite operator. They can be denoted as \begin{equation*}
    \mu_{(\bm{\lambda},\bm{\nu})}^{\text{LQ}}\left(\bm{A},\bm{B}\right) = \sqrt{\sigma_{\min}\left(LQ_{(\bm{\lambda},\bm{\nu})}\left(\bm{A},\bm{B}\right)\right) }
\end{equation*}
and
\begin{equation*}
    \mu_{(\bm{\lambda},\bm{\nu})}^{\text{Q}}\left(\bm{A},\bm{B}\right) = \sqrt{\sigma_{\min}\left(Q_{(\bm{\lambda},\bm{\nu})}\left(\bm{A},\bm{B}\right)\right) }.
\end{equation*}

From here we define the quadratic $\epsilon$-pseudospectrum.

\begin{definition}
Suppose that $\bm{A}$ in $\Pi_{i=1}^{d_1}\mathcal{M}_n(\C)$, $\bm{B}$ in $\Pi_{i=1}^{d_2}\mathcal{M}_n(\C)$ with $A_i$ Hermitian, $B_j$ non-Hermitian, and some $\epsilon \geq 0$ we define the non-Hermitian right, left, and quadratic \emph{$\epsilon$-pseudospectrum} as
    \begin{align*}
        \Lambda_\epsilon^{\text{RQ}} \left(\bm{A},\bm{B}\right) &= \left\{(\bm{\lambda},\bm{\nu}) \in \R^{d_1}\oplus \C^{d_2} \mid \mu_{(\bm{\lambda},\bm{\nu})}^{\text{RQ}}\left(\bm{A},\bm{B}\right)\leq \epsilon\right\} \\
        \Lambda_\epsilon^{\text{LQ}} \left(\bm{A},\bm{B}\right) &= \left\{(\bm{\lambda},\bm{\nu}) \in \R^{d_1}\oplus \C^{d_2} \mid \mu_{(\bm{\lambda},\bm{\nu})}^{\text{LQ}}\left(\bm{A},\bm{B}\right)\leq \epsilon\right\} \\
        \Lambda_\epsilon^{\text{Q}} \left(\bm{A},\bm{B}\right) &= \left\{(\bm{\lambda},\bm{\nu}) \in \R^{d_1}\oplus \C^{d_2} \mid \mu_{(\bm{\lambda},\bm{\nu})}^{\text{Q}}\left(\bm{A},\bm{B}\right)\leq \epsilon\right\} \\
    \end{align*}
\end{definition}

If we have $\epsilon =0$ in the definition of the $\epsilon$-pseudospectrum, there is a joint eigenvector $\bm{\psi}$ of $A_i$, and $B_j$ for all $i$ and $j$. This is more of a mathematical statement as we know that in physical systems, observables/operators do not commute and thus do not have a joint eigenspace.
From this point of view, the Proposition \ref{prop:nh_rquad_co_equiv_quantities} is a more useful result for physical systems when $\epsilon >0$.

In the study of physical systems we would like to see how the quadratic gap function behaves as we perturb our system as well as when we perturb the probe site.
When we treat the quadratic gap as a function of the probe site alone, we can think of the probe site as taking a guess at a point in the approximate joint spectrum of the matrices $A_i$ and $B_i$.
When we treat the quadratic gap as a function of the input matrices alone, we can think of studying perturbations to physical systems such as introducting point defects, line defects, etc.
Thus, a useful property of the right quadratic gap is Lipschitz continuity in $((\bm{A},\bm{B}),(\bm{\lambda},\bm{\nu}))$.
It also means that if we are running numerical simulations, determining a grid size for the probe sites can be predetermined ahead of time.
In the case where only the probe site varies, machine precision $\epsilon_{\text{mach}}$ implies that between grid spaces, the precision is also $\epsilon_{\text{mach}}$.

\begin{prop}
Suppose that $\bm{A},\bm{C}$ in $\Pi_{i=1}^{d_1}\mathcal{M}_n(\C), \bm{B}, \bm{D}$ in $\Pi_{i=1}^{d_2}\mathcal{M}_n(\C)$ with $A_i,C_i$ Hermitian, $B_j,D_j$ non-Hermitian, and $(\bm{\lambda},\bm{\nu}),(\bm{\alpha},\bm{\beta})$ in $\R^{d_1}\oplus \C^{d_2}$, then
\begin{equation*}
     \left\vert \mu_{(\bm{\lambda},\bm{\nu})}^{\text{RQ}}\left(\bm{A},\bm{B}\right) - \mu_{(\bm{\alpha},\bm{\beta})}^{\text{RQ}}\left(\bm{C},\bm{D}\right)\right\vert \leq \left\lVert(\bm{A},\bm{B})-(\bm{C},\bm{D})\right\rVert +\left\lVert(\bm{\lambda},\bm{\nu})-(\bm{\alpha},\bm{\beta})\right\rVert.
\end{equation*}
\end{prop}

\begin{proof}
    First note that
    \begin{align*}
        \left\lVert\begin{bmatrix} (\lambda_1 - \alpha_1)I \\ \vdots \\ (\lambda_{d_1} - \alpha_{d_1})I \\ (\nu_{1} - \beta_{1})I \\ \vdots \\ (\nu_{d_2} - \beta_{d_2})I\end{bmatrix} \right\rVert^2
        &=\left\lVert \sum_{i=1}^{d_1}\left\vert \lambda_i-\alpha_i\right\vert^2 I + \sum_{j=1}^{d_2}\left\vert \nu_j-\beta_j\right\vert^2 I\right\rVert \\
        &= \sum_{i=1}^{d_1}\left( \lambda_i-\alpha_i\right)^2 + \sum_{j=1}^{d_2}\left\vert \nu_j-\beta_j\right\vert^2 \\
        &= \vert\vert(\bm{\lambda},\bm{\nu})-(\bm{\alpha},\bm{\beta})\vert\vert^2.
    \end{align*}
    Then by a similar argument
    \begin{align*}
        \left\lVert\begin{bmatrix} (A_1 - C_1) \\ \vdots \\ (A_{d_1} - C_{d_1}) \\ (B_{1} - D_{1}) \\ \vdots \\ (B_{d_2} - D_{d_2})\end{bmatrix}\right\rVert^2
        &=\left\lVert \sum_{i=1}^{d_1}\left(A_i - C_i\right)^2 + \sum_{j=1}^{d_2}\left(B_{i} - D_{i}\right)^\dagger\left(B_{i} - D_{i}\right) \right\rVert \\
        &= \sum_{i=1}^{d_1}\left\lVert A_i-C_i \right\rVert^2 + \sum_{j=1}^{d_2}\left\lVert B_i-D_i \right\rVert^2 \\
        &= \vert\vert(\bm{A},\bm{B})-(\bm{C},\bm{D})\vert\vert^2.
    \end{align*}
It then follows that,
    \begin{align*}
        \left\lVert M_{(\bm{\lambda},\bm{\nu})} \left(\bm{A},\bm{B}\right)-M_{(\bm{\alpha},\bm{\beta})} \left(\bm{C},\bm{D}\right)\right\rVert
        &\leq \left\lVert\begin{bmatrix} (A_1 - C_1) \\ \vdots \\ (A_{d_1} - C_{d_1}) \\ (B_{1} - D_{1}) \\ \vdots \\ (B_{d_2} - D_{d_2})\end{bmatrix}\right\rVert + \left\lVert\begin{bmatrix} (\lambda_1 - \alpha_1)I \\ \vdots \\ (\lambda_{d_1} - \alpha_{d_1})I \\ (\nu_{1} - \beta_{1})I \\ \vdots \\ (\nu_{d_2} - \beta_{d_2})I\end{bmatrix} \right\rVert \\
        &\leq \vert\vert(\bm{A},\bm{B})-(\bm{C},\bm{D})\vert\vert +\vert\vert(\bm{\lambda},\bm{\nu})-(\bm{\alpha},\bm{\beta})\vert\vert.
    \end{align*}

    Lastly Proposition \ref{prop:nh_rquad_co_equiv_quantities} tells us

    \begin{align*}
        \left\vert \mu_{(\bm{\lambda},\bm{\nu})}^{\text{RQ}}\left(\bm{A},\bm{B}\right) - \mu_{(\bm{\alpha},\bm{\beta})}^{\text{RQ}}\left(\bm{E},\bm{F}\right)\right\vert &= \left|\sigma_{\min}\left( RM_{(\bm{\lambda},\bm{\nu})} \left(\bm{A},\bm{B}\right) \right)- \sigma_{\min}\left(RM_{(\bm{\alpha},\bm{\beta})} \left(\bm{E},\bm{F}\right)\right)\right|.
    \end{align*}
    Combine this with perturbation of singular values on page 78 of \cite{bhatia} and we are done.
\end{proof}

\begin{corollary}
Suppose that $\bm{A},\bm{C}$ in $\Pi_{i=1}^{d_1}\mathcal{M}_n(\C), \bm{B}, \bm{D}$ in $\Pi_{i=1}^{d_2}\mathcal{M}_n(\C)$ with $A_i,C_i$ Hermitian, $B_j,D_j$ non-Hermitian, and $(\bm{\lambda},\bm{\nu}),(\bm{\alpha},\bm{\beta})$ in $\R^{d_1}\oplus \C^{d_2}$, then
\begin{equation*}
     \left\vert \mu_{(\bm{\lambda},\bm{\nu})}^{\text{LQ}}\left(\bm{A},\bm{B}\right) - \mu_{(\bm{\alpha},\bm{\beta})}^{\text{LQ}}\left(\bm{C},\bm{D}\right)\right\vert \leq \left\lVert(\bm{A},\bm{B}^\dagger)-(\bm{C},\bm{D}^\dagger)\right\rVert +\left\lVert(\bm{\lambda},\bm{\nu}^*)-(\bm{\alpha},\bm{\beta}^*)\right\rVert.
\end{equation*}
With $\bm{B}^\dagger$ applying pointwise adjoint and $\bm{\nu}^*$ also conjugate applying pointwise to the entries.
\end{corollary}

\begin{proof}
    This follow from the fact that $LQ_{(\bm{\lambda},\bm{\nu})}\left(\bm{A},\bm{B}\right) = RQ_{(\bm{\lambda},\bm{\nu}^*)}\left(\bm{A},\bm{B}^\dagger\right)$
\end{proof}
These two results tell us that the quadratic gap function is pointwise continuous, but we conjecture that it might also be Lipschitz continuous.
\begin{corollary}
    Suppose that $\bm{A},\bm{C}$ in $\Pi_{i=1}^{d_1}\mathcal{M}_n(\C), \bm{B}, \bm{D}$ in $\Pi_{i=1}^{d_2}\mathcal{M}_n(\C)$ with $A_i,C_i$ Hermitian, $B_j,D_j$ non-Hermitian, and $(\bm{\lambda},\bm{\nu}),(\bm{\alpha},\bm{\beta})$ in $\R^{d_1}\oplus \C^{d_2}$, then the quadratic gap function is pointwise continuous as a function from $(\R^{d_1}\oplus\C^{d_2}) \times \Pi_{i=1}^{d_1+d_2}\mathcal{M}_n(\C)$ to $\R$.
\end{corollary}
\begin{proof}
    Lipschitz continuity of the the left and right quadratic operator gap functions as well as pointwise continuity of the minimum function.
\end{proof}

\section{Relationship between Clifford linear, radial and quadratic gaps} \label{sec:rel_spec_and_quad_gap}
As mentioned in Section \ref{sec:nh_clifford_composite_operator} we noted that the non-Hermitian spectral localizer is non-normal. This means that the variation in the spectrum can vary drastically.

Despite this we have that Theorem 1 in \cite{ruhe},
\begin{equation}\label{eq:min_spec-sigma_min}
    \left\vert\sigma_{\min}\left(L_{(\bm{\lambda},\nu)} \left( \bm{A}, B \right)\right) - \min\left|\textup{spec}\left(L_{(\bm{\lambda},\nu)} \left( \bm{A}, B \right)\right)\right|\right\vert \leq \Delta\left(L_{(\bm{\lambda},\nu)} \left( \bm{A}, B \right)\right).
\end{equation}

By claim (ii) in \cite{kahan}, equivalence of norms and by a geometric argument we have that
\begin{equation}\label{eq:min_real_spec-min_spec}
    \left\vert \min\left\vert \textrm{Re}\left(\textup{Spec}\left(L_{(\bm{\lambda},\nu)} \left( \bm{A}, B \right)\right)\right) \right\vert - \min\left\vert \textup{Spec}\left(L_{(\bm{\lambda},\nu)} \left( \bm{A}, B \right) \right) \right\vert \right\vert \leq \sqrt{n} \left\lVert L_{(\bm{\lambda},\nu)} \left( \bm{A}, B \right) - L_{(\bm{\lambda},\nu)} \left( \bm{A}, B \right)^\dagger\right\rVert
\end{equation}
where $n$ is the dimension of $L_{(\bm{\lambda},\nu)} \left( \bm{A}, B \right)$.

Lastly
\begin{equation}\label{eq:diff_l_and_adjoint}
    \left\lVert L_{(\bm{\lambda},\nu)} \left( \bm{A}, B \right) - L_{(\bm{\lambda},\nu)} \left( \bm{A}, B \right)^\dagger\right\rVert = \left\lVert (B-\nu) - (B-\nu)^\dagger \otimes I_m \right\rVert = \left\lVert (B-\nu) - (B-\nu)^\dagger \right\rVert.
\end{equation}

this will establish a relationship between the Clifford linear and radial gaps.

\begin{prop}\label{prop:nh_clg_linear_radial_gapped}
    Suppose that we have tuples of matrices $\bm{A}$ in $\Pi_{i=1}^{d_1}M_n(\C), B$ in $M_n(\C)$ with $A_i$ Hermitian, $B$ non-Hermitian, and $(\bm{\lambda},\nu)$ in $\R^{d_1}\oplus \C$, then
    \begin{equation*}
        \left\vert \bar{\mu}_{(\bm{\lambda},\nu)}^{\text{C}}\left(\bm{A},B\right) - \dot{\mu}_{(\bm{\lambda},\nu)}^{\text{C}}\left(\bm{A},B\right)\right\vert \leq \sqrt{N}\left\lVert (B-\nu)- (B-\nu)^\dagger\right\rVert + \Delta(L_{(\bm{\lambda},\nu)} \left( \bm{A}, B \right)),
    \end{equation*}
    where $N$ is the dimension of $L_{(\bm{\lambda},\nu)} \left( \bm{A}, B \right)$.
\end{prop}

\begin{proof}
    This follows from Equations (\ref{eq:min_spec-sigma_min}), (\ref{eq:min_real_spec-min_spec}) and (\ref{eq:diff_l_and_adjoint}).
\end{proof}

Indeed this formally tells us that if $\nu$ is real valued and $B$ is Hermitian, the two gaps are the same.
In fact $B$ Hermitian and $\nu$ real means we have equality with the spectral localizer gap in \cite{cerjan_loring_vides_2023}, and \cite{loring}.

\begin{prop}\label{prop:nh_clg_point_gapped}
    Suppose that we have tuples of matrices $\bm{A} \in \Pi_{i=1}^{d_1}M_n(\C), B$ in $M_n(\C)$ with $A_i$ Hermitian, $B$ non-Hermitian, and $(\bm{\lambda},\nu)$ in $\R^{d_1}\oplus \C$, then
    \begin{equation*}
        \left\vert \dot{\mu}_{(\bm{\lambda},\nu)}^{\text{C}}\left(\bm{A},B\right) - \mu_{(\bm{\lambda},\nu)}^{\text{Q}}\left(\bm{A},B\right)\right\vert \leq \sqrt{\sum_{i<j} \vert\vert[A_i,A_j]\vert\vert + \vert\vert F\vert\vert},
    \end{equation*}
    with $F$ as in equation \ref{eq:f_term}.
\end{prop}

\begin{proof}
    Note that tensoring by the identity does not affect the spectrum of a matrix (ignoring multiplicity). Therefore, given
    \begin{equation*}
    K = \sum_{i=1}^{d}(A_i-\lambda_i)^2 \otimes I + (B-\nu)^{\dagger}(B-\nu) \otimes \begin{bmatrix} I_m & \\ & 0_m \end{bmatrix}+ (B-\nu)(B-\nu)^{\dagger} \otimes \begin{bmatrix}0_m & \\  & I_m\end{bmatrix}
    \end{equation*}
    we get that
    \begin{equation*}
    \sqrt{\sigma_{\min}\left(Q_{(\bm{\lambda},\nu)}\left(\bm{A},B\right)\right)} = \sqrt{\sigma_{\min}\left( K \right)}.
    \end{equation*}
    We also have that
    \begin{align*}
        \left\lVert L_{(\bm{\lambda},\nu)} \left( \bm{A}, B \right)^{\dagger}L_{(\bm{\lambda},\nu)} \left( \bm{A}, B \right) -  K \right\rVert &= \left\lVert \sum_{i<j} [A_i,A_j] \otimes \Gamma_i\Gamma_j + F \right\rVert \\
        &\leq \left\lVert \sum_{i<j} [A_i,A_j] \otimes \Gamma_i\Gamma_j \right\rVert + \left\lVert F\right\rVert \\
        &\leq \sum_{i<j} \left\lVert  [A_i,A_j] \right\rVert + \left\lVert F\right\rVert.
    \end{align*}

    Next we have that for any positive quantities $a,b$ that $\left\vert\sqrt{a}-\sqrt{b}\right\vert\leq \sqrt{|a-b|}$.
    Combining this with Lipschitz continuity of singular values we get that
    \begin{align*}
        \left\vert \dot{\mu}_{(\bm{\lambda},\nu)}^{\text{C}}\left(\bm{A},B\right) - \mu_{(\bm{\lambda},\nu)}^{\text{Q}}\left(\bm{A},B\right) \right\vert &= \left\vert \sigma_{\min}(L_{(\bm{\lambda},\nu)}\left(\bm{A},B\right)) - \sqrt{\sigma_{\min}(K)} \right\vert \\
        &= \left\vert \sqrt{\sigma_{\min}\left((L_{(\bm{\lambda},\nu)}\left(\bm{A},B\right))^\dagger(L_{(\bm{\lambda},\nu)}\left(\bm{A},B\right))\right)} - \sqrt{\sigma_{\min}(K)} \right\vert \\
        &= \sqrt{\left\vert\sigma_{\min}\left((L_{(\bm{\lambda},\nu)}\left(\bm{A},B\right))^\dagger(L_{(\bm{\lambda},\nu)}\left(\bm{A},B\right))\right) - \sigma_{\min}(K)\right\vert} \\
        &\leq \sqrt{\sum_{i<j} \left\lVert[A_i,A_j]\right\rVert + \left\lVert F\right\rVert}.
    \end{align*}
\end{proof}

Combining the results from \ref{prop:nh_clg_linear_radial_gapped} and \ref{prop:nh_clg_point_gapped} we have the following proposition.

\begin{prop}\label{prop:nh_clg_line_gapped}
Suppose that we have tuples of matrices $\bm{A}$ in $\Pi_{i=1}^{d_1}M_n(\C), B$ in $M_n(\C)$ with $A_i$ Hermitian, $B$ non-Hermitian, and $(\bm{\lambda},\nu)$ in $\R^{d_1}\oplus \C$, then
\begin{align*}
    \left\vert \bar{\mu}_{(\bm{\lambda},\nu)}^{\text{C}}\left(\bm{A},B\right) - \mu_{(\bm{\lambda},\nu)}^{\text{Q}}\left(\bm{A},B\right)\right\vert &\leq \sqrt{n}\left\lVert (B-\nu) - (B-\nu)^\dagger \right\rVert + \Delta\left(L_{(\bm{\lambda},\nu)} \left( \bm{A}, B \right)\right) +\sqrt{\sum_{i<j} \left\lVert[A_i,A_j]\right\rVert + \left\lVert F\right\rVert},
\end{align*}
where $n$ is the dimension of $L_{(\bm{\lambda},\nu)} \left( \bm{A}, B \right)$ and $F$ as in equation \ref{eq:f_term}.
\end{prop}

\begin{proof}
    By Proposition \ref{prop:nh_clg_linear_radial_gapped} and \ref{prop:nh_clg_point_gapped} we have that,
    \begin{align*}
        \left\vert \bar{\mu}_{(\bm{\lambda},\nu)}^{\text{C}}\left(\bm{A},B\right) - \sigma_{\min}(K) \right\vert &\leq \left\vert \bar{\mu}_{(\bm{\lambda},\nu)}^{\text{C}}\left(\bm{A},B\right) - \sigma_{\min}\left(L_{(\bm{\lambda},\nu)} \left( \bm{A}, B \right)\right) \right\vert + \left\vert\sigma_{\min}\left(L_{(\bm{\lambda},\nu)} \left( \bm{A}, B \right)\right) - \sigma_{\min}(K) \right\vert \\
        &= \left\vert \bar{\mu}_{(\bm{\lambda},\bm{\nu})}^{\text{C}}\left(\bm{A},B\right) - \dot{\mu}_{(\bm{\lambda},\bm{\nu})}^{\text{C}}\left(\bm{A},B\right)\right\vert + \left\vert \dot{\mu}_{(\bm{\lambda},\bm{\nu})}^{\text{C}}\left(\bm{A},B\right) - \mu_{(\bm{\lambda},\bm{\nu})}^{\text{Q}}\left(\bm{A},B\right)\right\vert \\
        &\leq \sqrt{n}\left\lVert (B-\nu) - (B-\nu)^\dagger \right\rVert + \Delta\left(L_{(\bm{\lambda},\nu)} \left( \bm{A}, B \right)\right)+ \sqrt{\sum_{i<j} \left\lVert[A_i,A_j]\right\rVert + \left\lVert F\right\rVert}.
    \end{align*}
\end{proof}
Proposition \ref{prop:nh_clg_point_gapped} is equivalent Proposition II.4 in Ref \cite{cerjan_loring_vides_2023} when the the operator $B$ is Hermitian and the probe site $\nu$ is restricted to be real valued.

\section{Local nature of the quadratic composite operator}\label{sec:local_nature_nh_quadratic}

The non-Hermitian composite operator being local in this setting means that if we had a physical system, a perturbation to the Hamiltonian far away from the probe site $(\bm{\lambda},\bm{\nu})$ will have little to no effect on the quadratic gap.
In other words, small defects away from a probe site should not have significantly large effects when probing away from the defect.
The following theorem establishes this idea.
In fact this is a generalization of Theorem III.1 from \cite{cerjan_loring_vides_2023}, and the proof is in the same spirit.

For the remainder of the section we will assume that the entries of the matrix $d_1$ tuple in $\bm{A}$ commute and that each of the $A_i$ is invertible.
In a physical setting where $A_i = X_i$ are position observables this is easily accomplished by shifting the system in position so that the corresponding operator is invertible.

This allows us to consider the following construction. It can be thought of as the Euclidean distance from the origin in spacial coordinates.
Specifically,
\begin{equation*}
    Z^2 = \sum_{i=1}^{d_1}(A_i-\lambda_i)^2
\end{equation*}
and
\begin{equation*}
    Y^2 = \sum_{j=1}^{d_2}(B_j-\nu_j)^\dagger (B_j-\nu_j).
\end{equation*}

\begin{theorem}
    Assume that $\bm{A}$ in $\Pi_{i=1}^{d_1}\mathcal{M}_n(\C)$ is a $d_1$-tuple of Hermitian matrices, $\bm{B},\bm{C}$ in $\Pi_{i=1}^{d_2}\mathcal{M}_n(\C)$ are $d_2$-tuples of non-Hermitian matrices then we have that if
    \begin{equation*}
        \left\lVert Z^{-1}\left( \sum_{j=1}^{d_2}(B_j-\nu_j)^\dagger C_j + C_j^\dagger (B_j-\nu_j) +C_j^\dagger C_j\right)Z^{-1}\right\rVert\leq K
    \end{equation*}
    for some $K<1$, then
    \begin{equation*}
        (1-K)^{1/2}\mu_{(\bm{\lambda},\bm{\nu})}^\text{RQ}\left(\bm{A},\bm{B}\right) \leq \mu_{(\bm{\lambda},\bm{\nu})}^\text{RQ}\left(\bm{A},\bm{B}+\bm{C}\right) \leq (1+K)^{1/2}\mu_{(\bm{\lambda},\bm{\nu})}^\text{RQ}\left(\bm{A},\bm{B}\right).
    \end{equation*}
\end{theorem}

\begin{proof}
    After direct computation we have that
    \begin{equation*}
        -KZ^2\leq \left( \sum_{j=1}^{d_2}(B_j-\nu_j)^\dagger C_j + C_j^\dagger (B_j-\nu_j) +C_j^\dagger C_j\right) \leq KZ^2.
    \end{equation*}
    Thus,
    \begin{align*}
        \text{RQ}_{(\bm{\lambda},\bm{\nu})}\left(\bm{A},\bm{B}+\bm{C}\right) &\leq Z^2 + Y^2 + KZ^2 \\
        &= (1+K)\left(Z^2 + Y^2\right) \\
        &= (1+K)\text{RQ}_{(\bm{\lambda},\bm{\nu})}\left(\bm{A},\bm{B}\right) \\
        \text{RQ}_{(\bm{\lambda},\bm{\nu})}\left(\bm{A},\bm{B}+\bm{C}\right) &\geq Z^2 + Y^2 - KZ^2 \\
        &\geq (1-K)\left(Z^2 + Y^2\right) \\
        &= (1-K)\text{RQ}_{(\bm{\lambda},\bm{\nu})}\left(\bm{A},\bm{B}\right).
    \end{align*}
    From here we have that
    \begin{equation*}
        (1+K)^{-1}\text{RQ}_{(\bm{\lambda},\bm{\nu})}\left(\bm{A},\bm{B}\right)^{-1} \leq \text{RQ}_{(\bm{\lambda},\bm{\nu})}\left(\bm{A},\bm{B}+\bm{C}\right)^{-1} \leq (1-K)^{-1}\text{RQ}_{(\bm{\lambda},\bm{\nu})}\left(\bm{A},\bm{B}\right)^{-1}.
    \end{equation*}
    Then since the quadratic composite operator is Hermitian Positive Semidefinite we have that
    \begin{equation*}
        (1+K)^{-1}\left\lVert\text{RQ}_{(\bm{\lambda},\bm{\nu})}\left(\bm{A},\bm{B}\right)^{-1} \right\rVert \leq \left\lVert\text{RQ}_{(\bm{\lambda},\bm{\nu})}\left(\bm{A},\bm{B}+\bm{C}\right)^{-1} \right\rVert \leq (1-K)^{-1}\left\lVert\text{RQ}_{(\bm{\lambda},\bm{\nu})}\left(\bm{A},\bm{B}\right)^{-1}\right\rVert.
    \end{equation*}

    Finally,
    \begin{equation*}
        (1-K)^{\frac{1}{2}}\left\lVert\text{RQ}_{(\bm{\lambda},\bm{\nu})}\left(\bm{A},\bm{B}\right)^{-1} \right\rVert^{-\frac{1}{2}} \leq \left\lVert\text{RQ}_{(\bm{\lambda},\bm{\nu})}\left(\bm{A},\bm{B}+\bm{C}\right)^{-1} \right\rVert^{-\frac{1}{2}} \leq (1+K)^{\frac{1}{2}}\left\lVert\text{RQ}_{(\bm{\lambda},\bm{\nu})}\left(\bm{A},\bm{B}\right)^{-1}\right\rVert^{-\frac{1}{2}}.
    \end{equation*}
\end{proof}

Note that in the proof we make extensive use that for Hermitian matrices $A,B$, $A\geq B$ if and only if $A-B$ is positive semidefinite.
This generates an ordering on the set of positive semidefinite matrices.

\begin{remark}\label{remark:locality_left_q_gap}
    We can also show that the left quadratic gap function is local by the same argument.
\end{remark}

\begin{corollary}
        Assume that $\bm{A}$ in $\Pi_{i=1}^{d_1}\mathcal{M}_n(\C)$ is a $d_1$-tuple of Hermitian matrices, $\bm{B},\bm{C}$ in $\Pi_{i=1}^{d_2}\mathcal{M}_n(\C)$ are $d_2$-tuples of non-Hermitian matrices then we have that if
    \begin{equation*}
        \left\lVert Z^{-1}\left( \sum_{j=1}^{d_2}(B_j-\nu_j)^\dagger C_j + C_j^\dagger (B_j-\nu_j) +C_j^\dagger C_j\right)Z^{-1}\right\rVert\leq K
    \end{equation*}
    for some $K<1$, then
    \begin{equation*}
        (1-K)^{1/2}\mu_{(\bm{\lambda},\bm{\nu})}^\text{Q}\left(\bm{A},\bm{B}\right) \leq \mu_{(\bm{\lambda},\bm{\nu})}^\text{Q}\left(\bm{A},\bm{B}+\bm{C}\right) \leq (1+K)^{1/2}\mu_{(\bm{\lambda},\bm{\nu})}^\text{Q}\left(\bm{A},\bm{B}\right).
    \end{equation*}
\end{corollary}

\begin{proof}
    Combine Theorem \ref{remark:locality_left_q_gap}, Remark \ref{remark:locality_left_q_gap} and the fact that the minimum function preserves inequalities.
\end{proof}

\section{Applications to a point gaped system}\label{sec:applications_to_point_gapped_systems}

In order to encounter non-Hermiticity in physical systems one need not go very far.
Indeed we consider a two level system such as the one mentioned in \cite{tls_paper}.
The two level system gives rise to a non-Hermitian Hamiltonian operator.
In particular we have a point gap system.
We will provide some numerical justification for why the Clifford radial gap might be the best one to use over the established Clifford linear gap.
To illustrate, we first consider a general case with shifted energy as follows:
\begin{equation*}
    X = \begin{bmatrix} -1 & 0 \\ 0 & 1 \end{bmatrix} \mbox{ and } H = \begin{bmatrix} \Delta E + i\Delta\gamma & c \\ c & 0 \end{bmatrix}
\end{equation*}
with $\Delta E, \Delta\gamma, c \in \R$.
By doing a traditional eigenenergy and eigenvector analysis we have that the eigenenergy and eigenvector pair are
\begin{equation*}
    \lambda_{\pm}(H) = \frac{(\Delta E + i\Delta\gamma)\pm\sqrt{(\Delta E + i\Delta\gamma)^2+4c^2}}{2} \hspace{0.1in},\hspace{0.1in}
    \bm{\psi}_{\pm} = \begin{bmatrix} \lambda_{\pm}(H)/c & 1 \end{bmatrix}^T
\end{equation*}

If we consider the case when $\Delta E=0$ then we have an exceptional point when $c=\mp \Delta\gamma/2$.
In this scenario we have the eigenenergy and eigenvector pair
\begin{equation*}
    \lambda_{\pm} = \frac{i\Delta\gamma}{2} \hspace{0.1in},\hspace{0.1in} \bm{\psi} = \begin{bmatrix} i\Delta\gamma/c & 1 \end{bmatrix}^T
\end{equation*}
In other words the eigenspace of $H$ becomes degenerate!
The quadratic gap for this setup can get complicated, so we compute one of the exceptional point cases when $1=c=\Delta\gamma/2$.
We will fix the position probe site of $Q$ to be $x=0$ so that the non-Hermitian Clifford and quadratic composite operators become
\begin{align*}
    L_{(0,E)}(X,H) &= \begin{bmatrix}
        2i - E & 1 & -1 & 0 \\
        1 & -E & 0 & 1 \\
        -1 & 0 &-2i+E^* & -1 \\
        0 & 1 & -1 & +E^* \\
    \end{bmatrix} \\
    Q_{(0,E)}(X,H) &= \begin{bmatrix}
        \left\vert E\right\vert^2 + 6  + 4 i\textrm{Im}(E) & -2\textrm{Re}(E) -2i & & \\
        -2\textrm{Re}(E) + 2i & 2 + \left\vert E\right\vert^2 & & \\
        & & \left\vert E\right\vert^2 + 6 + 4 i\textrm{Im}(E) & -2\textrm{Re}(E) +2i \\
        & &-2\textrm{Re}(E) - 2i & 2 + \left\vert E\right\vert^2
    \end{bmatrix}
\end{align*}
In Table \ref{tbl:tls_gaps_and_diff} we compute the Clifford linear gap, radial gap, the quadratic gap, as well as the differences between each gap.
We can see that the Clifford radial gap (Figure \ref{fig:tls_radial_gap}) detects the exceptional point at probe site $E=i$ and agrees well with the quadratic gap (Figure \ref{fig:tls_quadratic_gap}) near the exceptional point well. This can be seen by looking at the difference of the two gaps in Figure \ref{fig:tls_radial-quadratic_gap}.
On the other hand, while the Clifford linear gap agrees with the quadratic gap near the exceptional point (Figure \ref{fig:tls_linear-quadratic_gap}), by itself it does not detect well the exceptional point (Figure \ref{fig:tls_linear_gap}).
This is the motivation for defining the \emph{Clifford radial gap}.

\begin{table}[ht!]
    \centering
    \begin{tabular}{ | c | c | c | }
    \hline \hline

    \begin{minipage}{0.3\textwidth}
    \captionof{figure}{\\$\bar{\mu}_{(0,E)}^{\text{C}}\left(X,H\right)$}\label{fig:tls_linear_gap}
    \includegraphics[width=\textwidth]{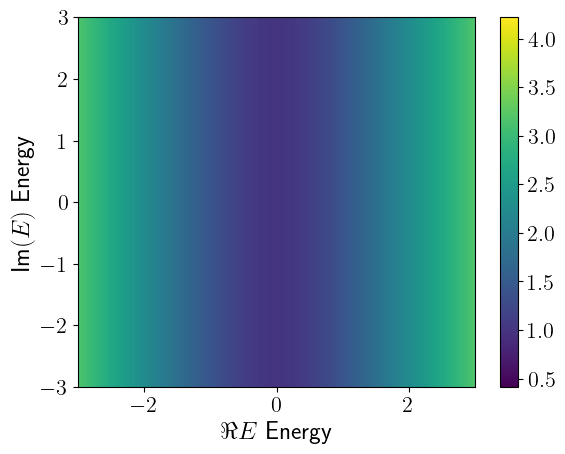}
    \end{minipage} &
    \begin{minipage}{0.3\textwidth}
    \captionof{figure}{\\$\dot{\mu}_{(0,E)}^{\text{C}}\left(X,H\right)$}\label{fig:tls_radial_gap}
    \includegraphics[width=\textwidth]{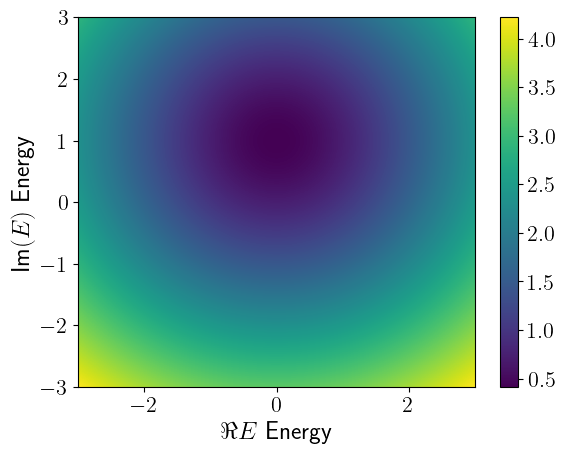}
    \end{minipage} &  \begin{minipage}{0.3\textwidth}
    \captionof{figure}{\\$\mu_{(0,E)}^{\text{Q}}\left(X,H\right)$}\label{fig:tls_quadratic_gap}
    \includegraphics[width=\textwidth]{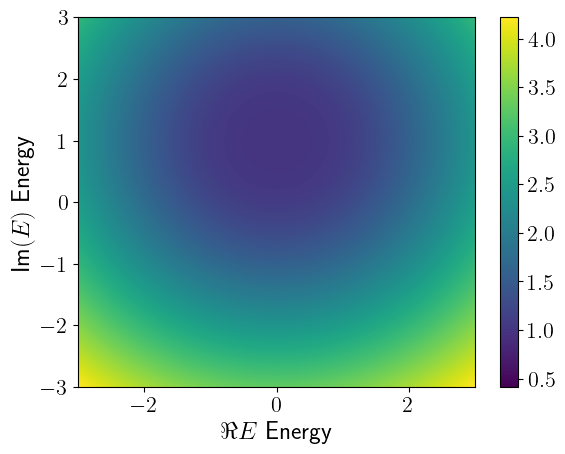}
    \end{minipage} \\ \hline \hline

    \begin{minipage}{0.3\textwidth}
    \captionof{figure}{\\$\left\vert \bar{\mu}_{(0,E)}^{\text{C}}\left(X,H\right) - \dot{\mu}_{(0,E)}^{\text{C}}\left(X,H\right) \right\vert$}\label{fig:tls_linear-radial_gap}
    \includegraphics[width=\textwidth]{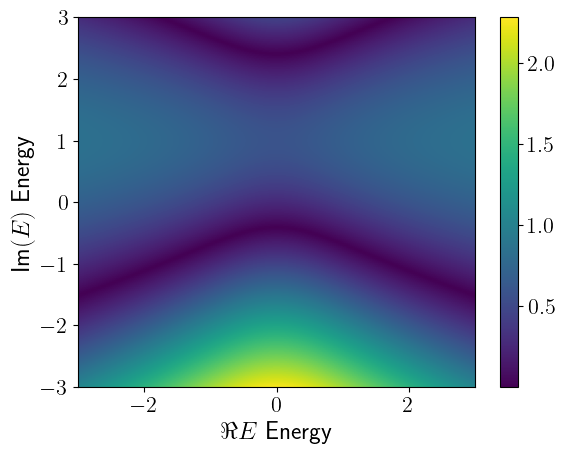}
    \end{minipage} &
    \begin{minipage}{0.3\textwidth}
    \captionof{figure}{\\$\left\vert \bar{\mu}_{(0,E)}^{\text{C}}\left(X,H\right) - \mu_{(0,E)}^{\text{Q}}\left(X,H\right) \right\vert$}\label{fig:tls_linear-quadratic_gap}
    \includegraphics[width=\textwidth]{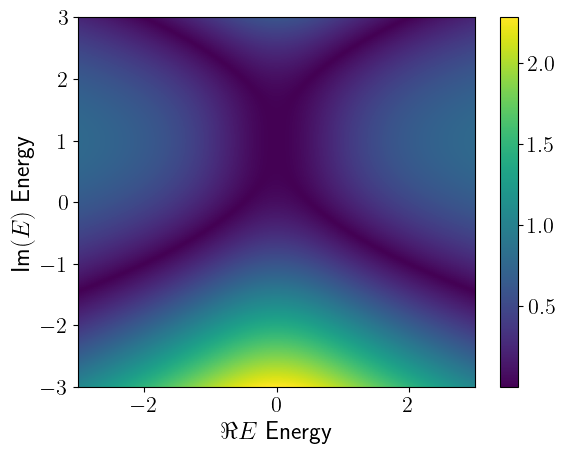}
    \end{minipage} &  \begin{minipage}{0.3\textwidth}
    \captionof{figure}{\\$\left\vert \dot{\mu}_{(0,E)}^{\text{C}}\left(X,H\right) - \mu_{(0,E)}^{\text{Q}}\left(X,H\right) \right\vert$}\label{fig:tls_radial-quadratic_gap}
    \includegraphics[width=\textwidth]{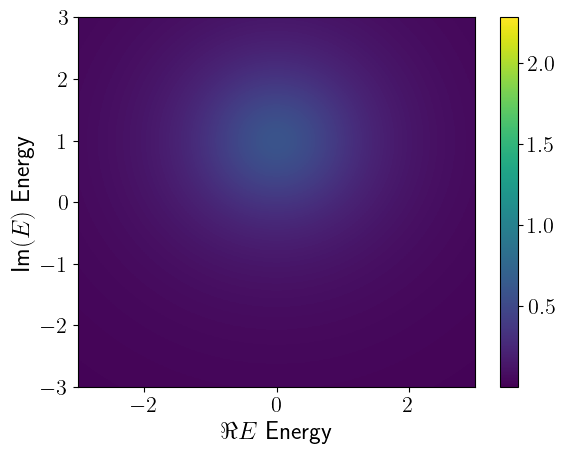}
    \end{minipage} \\ \hline \hline
    \end{tabular}
    \caption{Shown are the Clifford linear and radial gaps (Figure \ref{fig:tls_linear_gap} and \ref{fig:tls_radial_gap}), quadratic gap (Figure \ref{fig:tls_quadratic_gap}) as well as the difference between the gaps in Figure \ref{fig:tls_linear-radial_gap} \ref{fig:tls_linear-quadratic_gap}, \ref{fig:tls_radial-quadratic_gap} of a Two Level System with $c=1$,$\Delta\gamma = 2$ and position probe site $x=0$.} \label{tbl:tls_gaps_and_diff}
\end{table}

\section{Applications to a line gapped system}\label{sec:HHS_Lossy_Modes}

Now we will provide an example of a physical system where the non-Hermiticity comes from a non-Hermitian Hamiltonian which is line gapped.
Consider the generalized Haldane model over a bi-partite honeycomb lattice given in \cite{haldane} with the following tight binding model hamiltonian;

\begin{align*}
    H &= \sum_{n_A,n_B} \left( (M-i\mu)|n_A\rangle \langle n_A| - (M+i\mu)|n_B\rangle \langle n_B| \right) - t\sum_{\langle n_A,m_B \rangle} \left( |n_A\rangle \langle m_B| - |m_B\rangle \langle n_A| \right) \\
    &\qquad- t_c\sum_{\alpha=A,B}\sum_{\langle\langle n_\alpha,m_\alpha \rangle\rangle} \left( e^{i \phi(n_\alpha,m_\alpha)}| n_\alpha\rangle \langle m_\alpha | + e^{-i \phi(n_\alpha,m_\alpha)}| m_\alpha\rangle \langle n_\alpha |\right).
\end{align*}

The first sum runs over all the lattice sites with $A$ and $B$ lattices having on-site energies $\pm M$. The second sum is kinetic energy with nearest neighbor coupling coefficient $t$.
Lastly we have that the third sum is over next-nearest-neighbor pairs and has a direction-dependent phase factor $e^{i \phi(n_\alpha,m_\alpha)}$ a geometrically determined sign.
When $\mu=0$ in the Hamiltonian, we have no lossy sites and the Hamiltonian is Hermitian. By making the model lossy, the Hamiltonian becomes non-Hermitian.

Our example tight binding model will be comprised of an inner topological insulator surrounded by a ring of trivial insulator which is then surrounded by a ring of a lossy trivial insulator, we call it the Haldane Heterostructure as in \cite{cerjan_koekebier_Schulz_Baldes}.
Again the loss comes from having on site energies being $\pm M - i\mu$ on the lossy Trivial insulator sites.
The diagram for the tight binding model and unit cell is shown in Table \ref{tbl:hhs}.
Furthermore, we will focus our attention on fixing the parameters for the tight binding model to the ones listed in Table \ref{tbl:hhs_parameters}.

\begin{table}[H]
    \centering
    \begin{tabular}{ | c | c | }
    \hline \hline
    \begin{minipage}{0.4\textwidth}
    \captionof{figure}{Haldane Heterostructure Overview}\label{fig:hhs_overview}
    \includegraphics[width=0.95\textwidth]{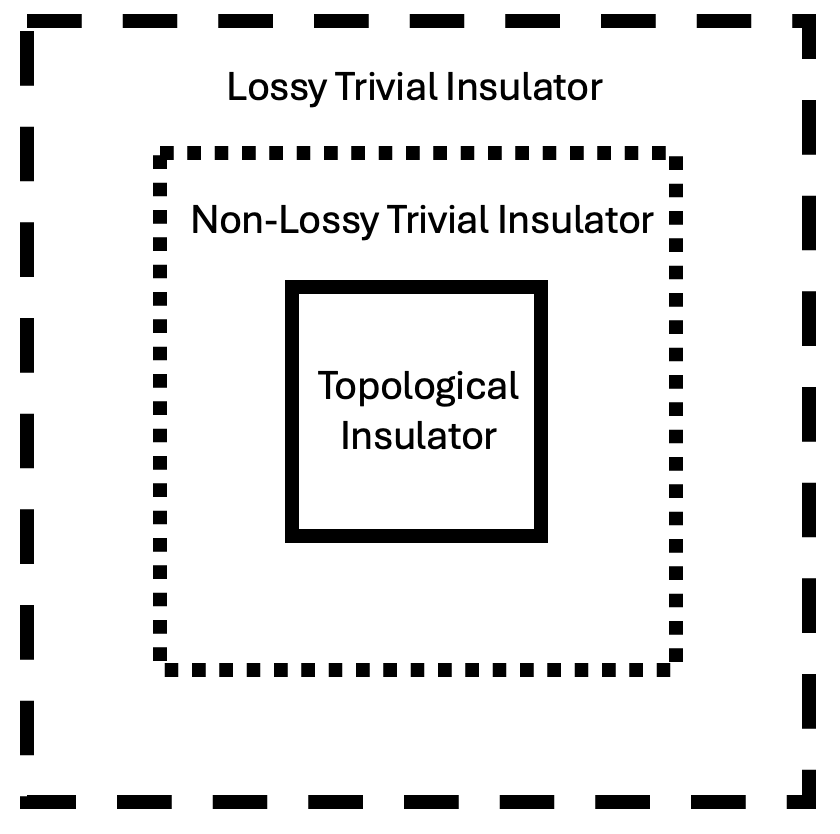}
    \end{minipage} &
    \begin{minipage}{0.4\textwidth}
    \captionof{figure}{Haldane Hetero Structure Cell}\label{fig:hhs_cell}
    \includegraphics[width=\textwidth]{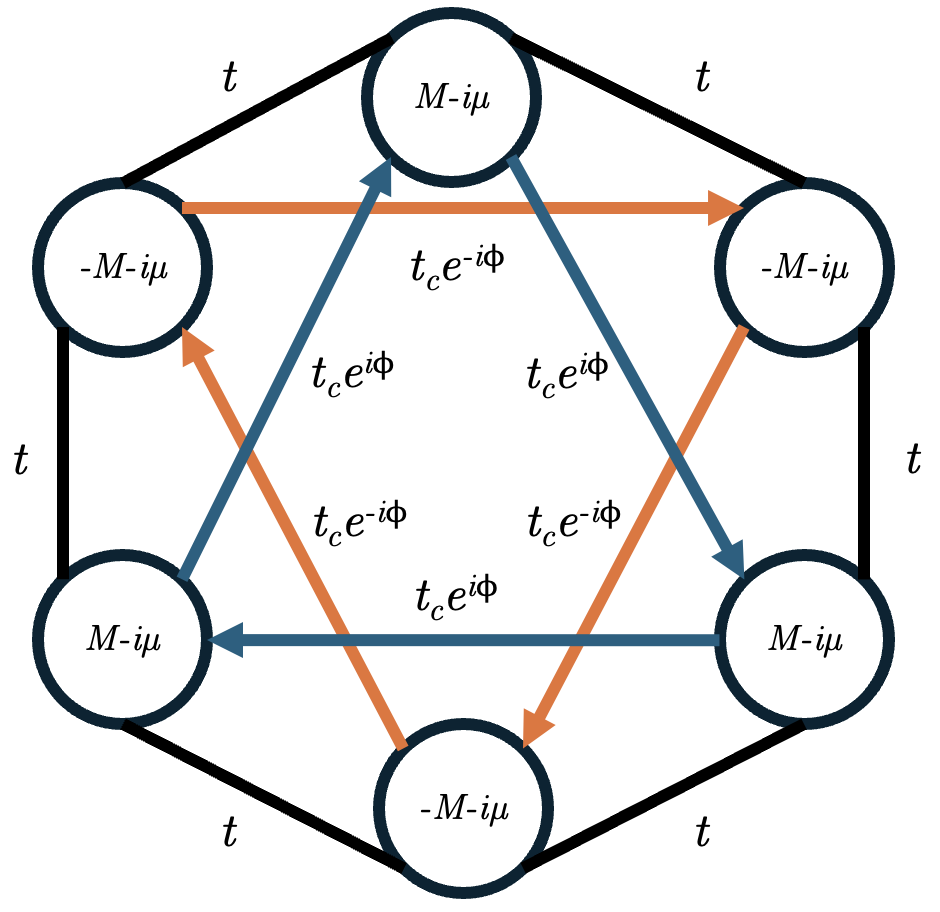}
    \end{minipage}\\ \hline \hline
    \end{tabular}
    \caption{Shown in Figure \ref{fig:hhs_overview} is the Haldane Heterostructure overview. Figure \ref{fig:hhs_cell} shows the general cell for the Haldane Heterostructure} \label{tbl:hhs}
\end{table}

\begin{table}[H]
    \centering
    \begin{tabular}{ c || c | c | c | c | c}
    Layer & $M$ & $\mu$ & $t$ & $t_c$ & $\phi$ \\ \hline \hline
    Trivial Insulator (Lossy) & $0.5\sqrt{3}$ & $0.2$ & $1$ & $0$ & N/A \\
    Trivial Insulator (Non-Lossy) & $0.5\sqrt{3}$ & $0$ & $1$ & $0$ & N/A \\
    Topological Insulator & $0$ & $0$ & $1$ & $0.5$ & $\frac{\pi}{2}$
    \end{tabular}
    \caption{Parameters for Heterostructure}\label{tbl:hhs_parameters}
\end{table}

In both the spectral localizer and quadratic composite operator we will use $\kappa = 0.5$. 
For the two positions operators we will take unit length between `A' sites.
They will end up being diagonal matrices as shown in Equation (\ref{eq:hhs_positions}).
\begin{equation}\label{eq:hhs_positions}
    X = \text{Diag}(x_1,...,x_i,...,x_n),
    \quad
    Y = 
    \text{Diag}(y_1,...,y_n)
\end{equation}
The Hamiltonian will be non-Hermitian, and sparse for large system size. In some basis the diagonal entries will consist of complex ($\pm M-i\mu$), purely real $\pm M$ and zeros as seen in Equation \ref{eq:hhs_hamiltonian}. In Equation \ref{eq:hhs_hamiltonian}, $C$ is a matrix with zeros on the diagonal.
\begin{equation}\label{eq:hhs_hamiltonian}
    H = \text{Diag}(M-i\mu, ..., -M-i\mu, ...,M,...,-M,...,0,...,0) + C
\end{equation}
The 2d non-Hermitian spectral localizer and the quadratic composite operator are as follows
\begin{align*}
    L^{\textup{2d}}_{(\kappa x,\kappa y,E)}(\kappa X,\kappa Y,H) &= \kappa(X-x)\otimes \sigma_x + \kappa(Y-y)\otimes \sigma_y + (H-E)\otimes \begin{bmatrix} 1 & \\ & 0 \end{bmatrix} + (H-E)^\dagger \otimes \begin{bmatrix} 0 & \\ & -1 \end{bmatrix}\\
    Q^{\textup{2d}}_{(\kappa x,\kappa y,E)}(\kappa X,\kappa Y,H) &= \left[\kappa\left(X-x\right)^2 + \kappa\left(Y-y\right)^2 \right] \otimes \begin{bmatrix}1 & \\ & 1 \end{bmatrix} \\
    &\qquad + (H-E)^{\dagger}(H-E) \otimes \begin{bmatrix} 1 & \\ & 0 \end{bmatrix} + (H-E)(H-E)^\dagger \otimes \begin{bmatrix} 0 & \\ & -1 \end{bmatrix}.
\end{align*}
The pseudospectra tell us where we should expect to find approximate joint states with the corresponding imaginary probe site $E$ telling us how lossy the state is.
We demonstrate via computation the Clifford linear and quadratic gap functions with varying position probe sites and fixed energy probe site in Table \ref{tbl:hhs_gaps_and_diff}.
In the same table we also demonstrate the differences between the various gap functions.

\begin{table}[ht!]
    \centering
    \begin{tabular}{ | c | c | c | }
    \hline \hline

    \begin{minipage}{0.3\textwidth}
    \captionof{figure}{\\$\bar{\mu}_{(\kappa\bm{x},0)}^{\text{C}}\left(\kappa\bm{X},H\right)$}\label{fig:hhs_linear_gap}
    \includegraphics[width=\textwidth]{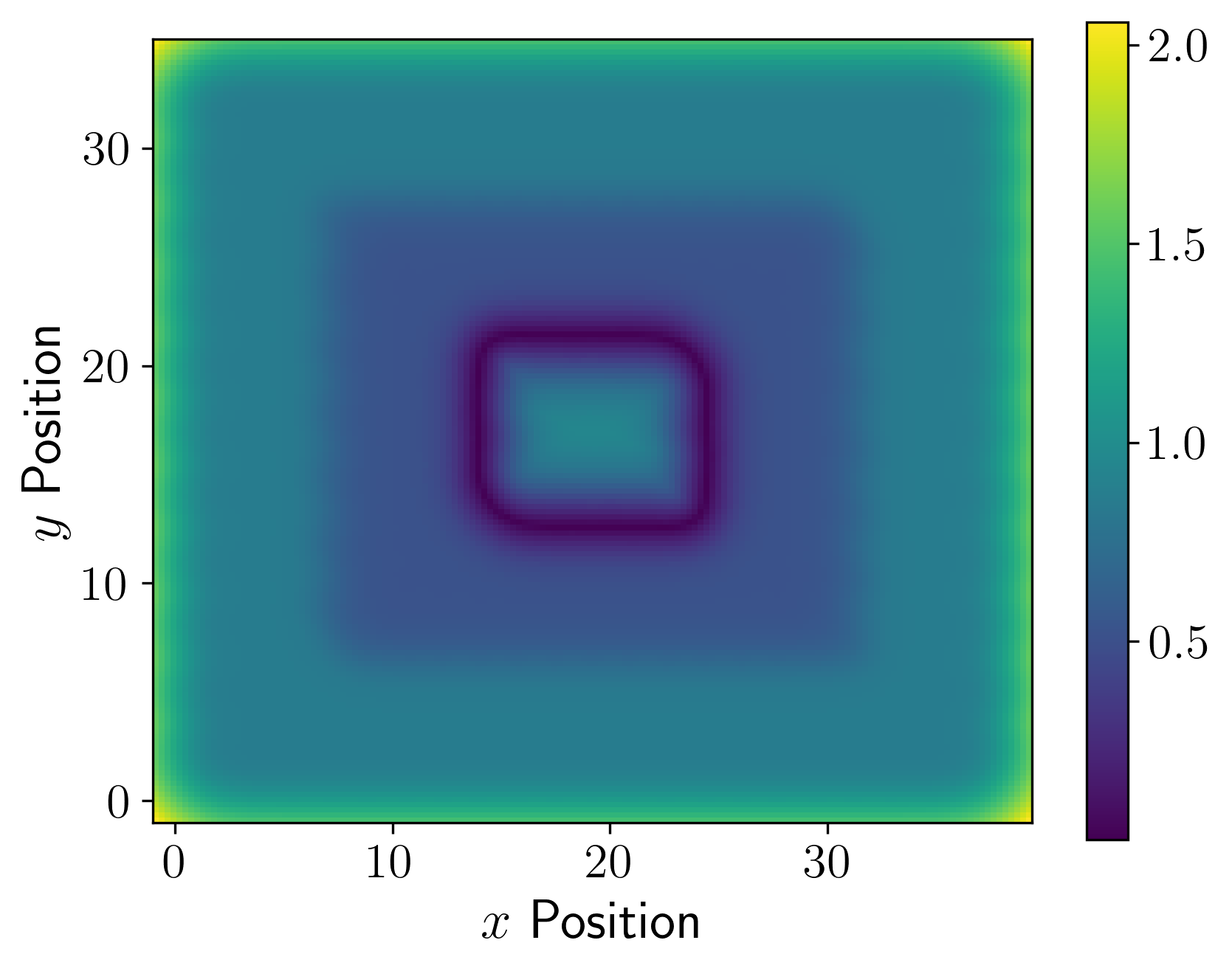}
    \end{minipage} &
    \begin{minipage}{0.3\textwidth}
    \captionof{figure}{\\$\mu_{(\kappa\bm{x},0)}^{\text{RQ}}\left(\kappa\bm{X},H\right)$}\label{fig:hhs_quadratic_gap}
    \includegraphics[width=\textwidth]{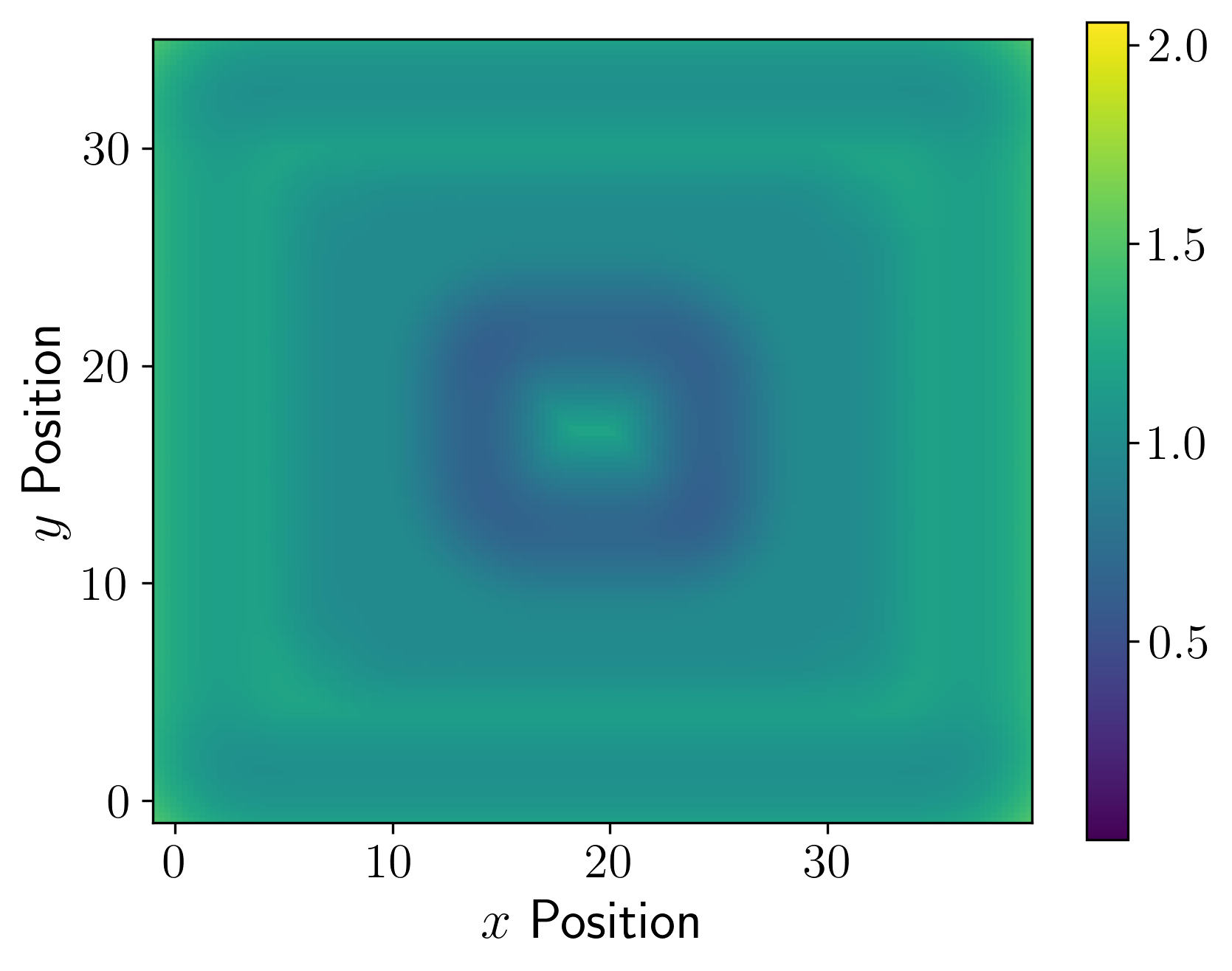}
    \end{minipage} &  \begin{minipage}{0.3\textwidth}
    \captionof{figure}{\\$\left\vert\bar{\mu}_{(\kappa\bm{x},0)}^{\text{C}}\left(\kappa\bm{X},H\right) - \mu_{(\kappa\bm{x},0)}^{\text{RQ}}\left(\kappa\bm{X},H\right)\right\vert$}\label{fig:hhs_gap_diff}
    \includegraphics[width=\textwidth]{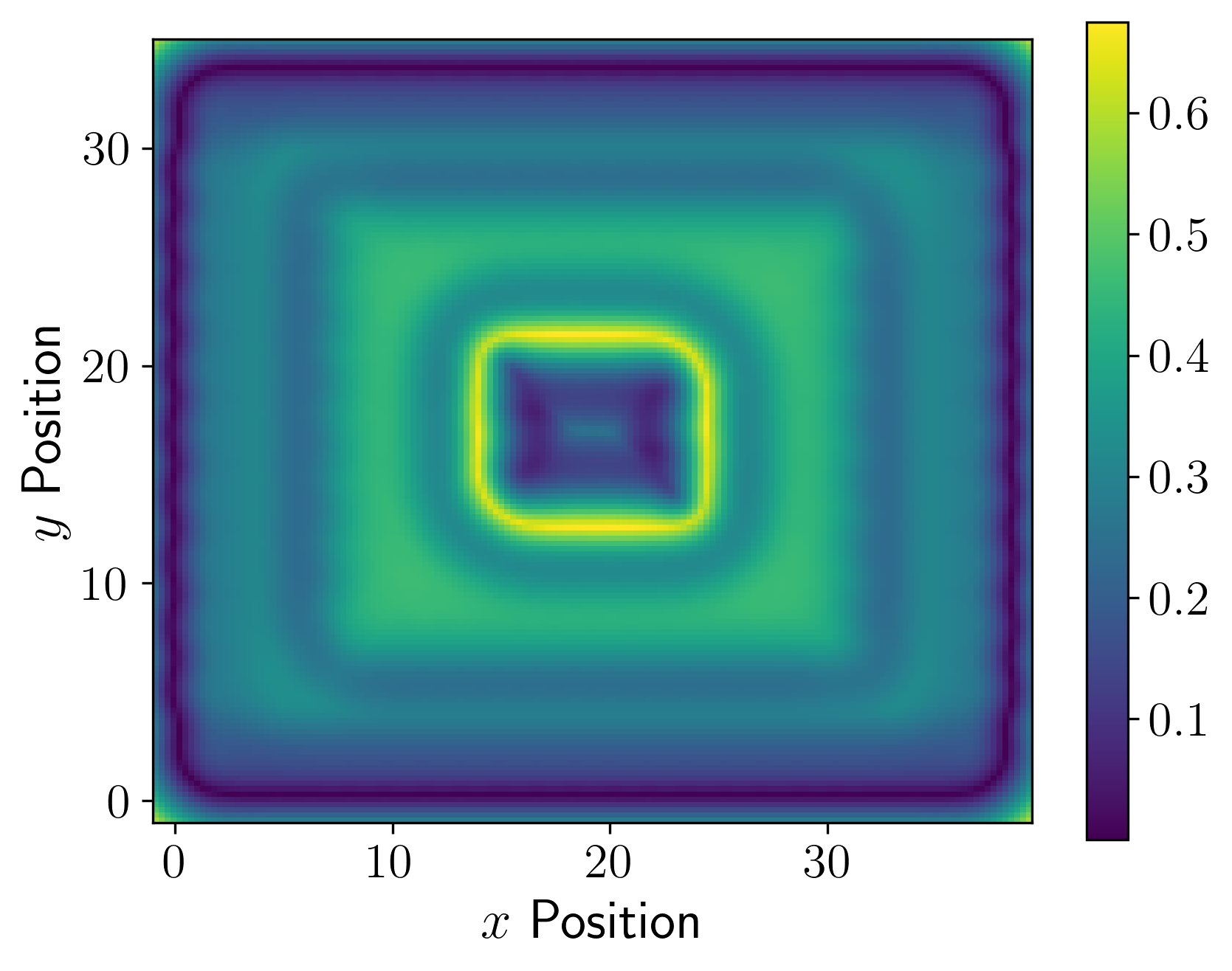}
    \end{minipage} \\
    \hline \hline

    \begin{minipage}{0.3\textwidth}
    \captionof{figure}{\\$\bar{\mu}_{(\kappa\bm{x},1)}^{\text{C}}\left(\kappa\bm{X},H\right)$}\label{fig:hhs_linear_gap_e_1}
    \includegraphics[width=\textwidth]{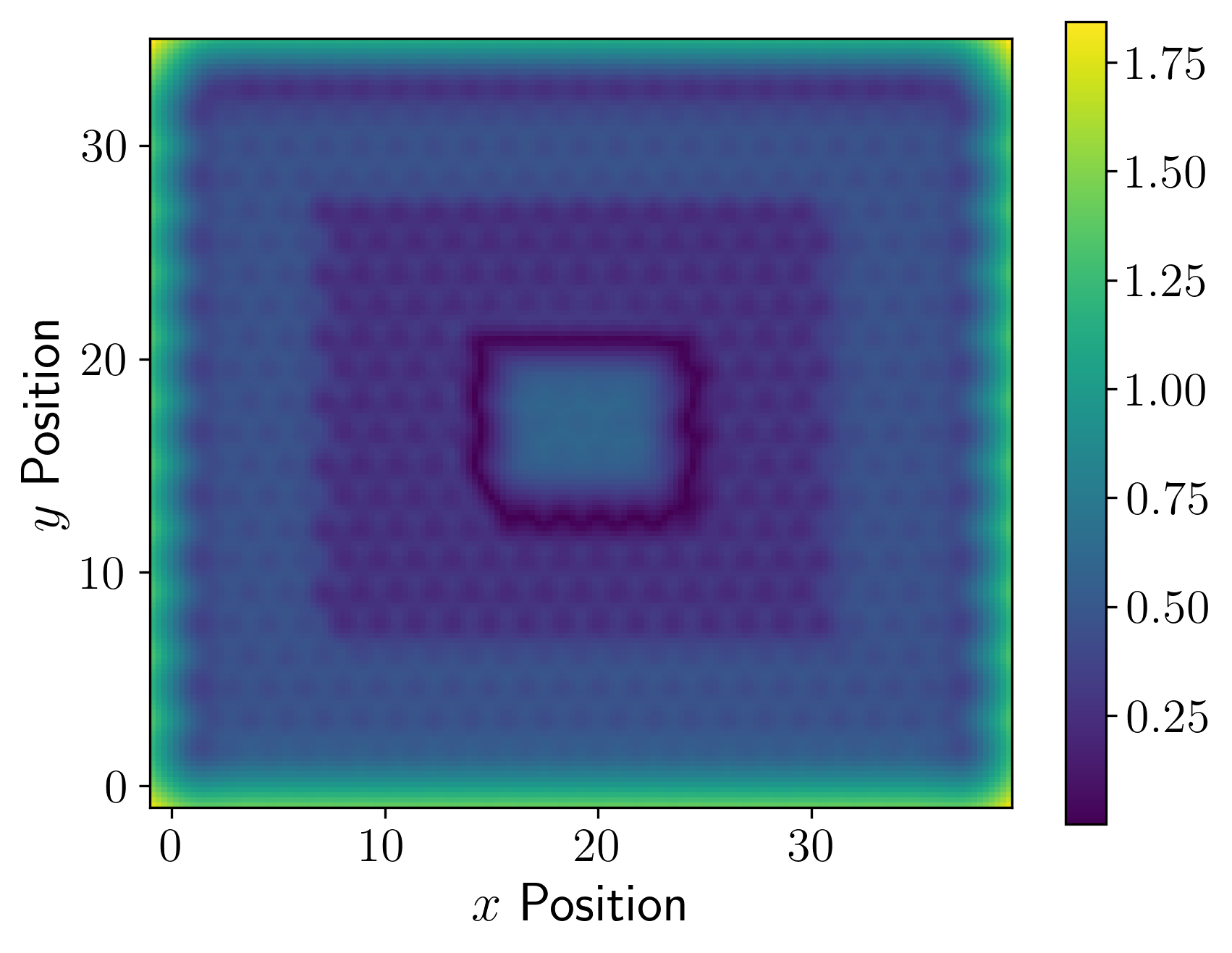}
    \end{minipage} &
    \begin{minipage}{0.3\textwidth}
    \captionof{figure}{\\$\mu_{(\kappa\bm{x},1)}^{\text{RQ}}\left(\kappa\bm{X},H\right)$}\label{fig:hhs_quadratic_gap_e_1}
    \includegraphics[width=\textwidth]{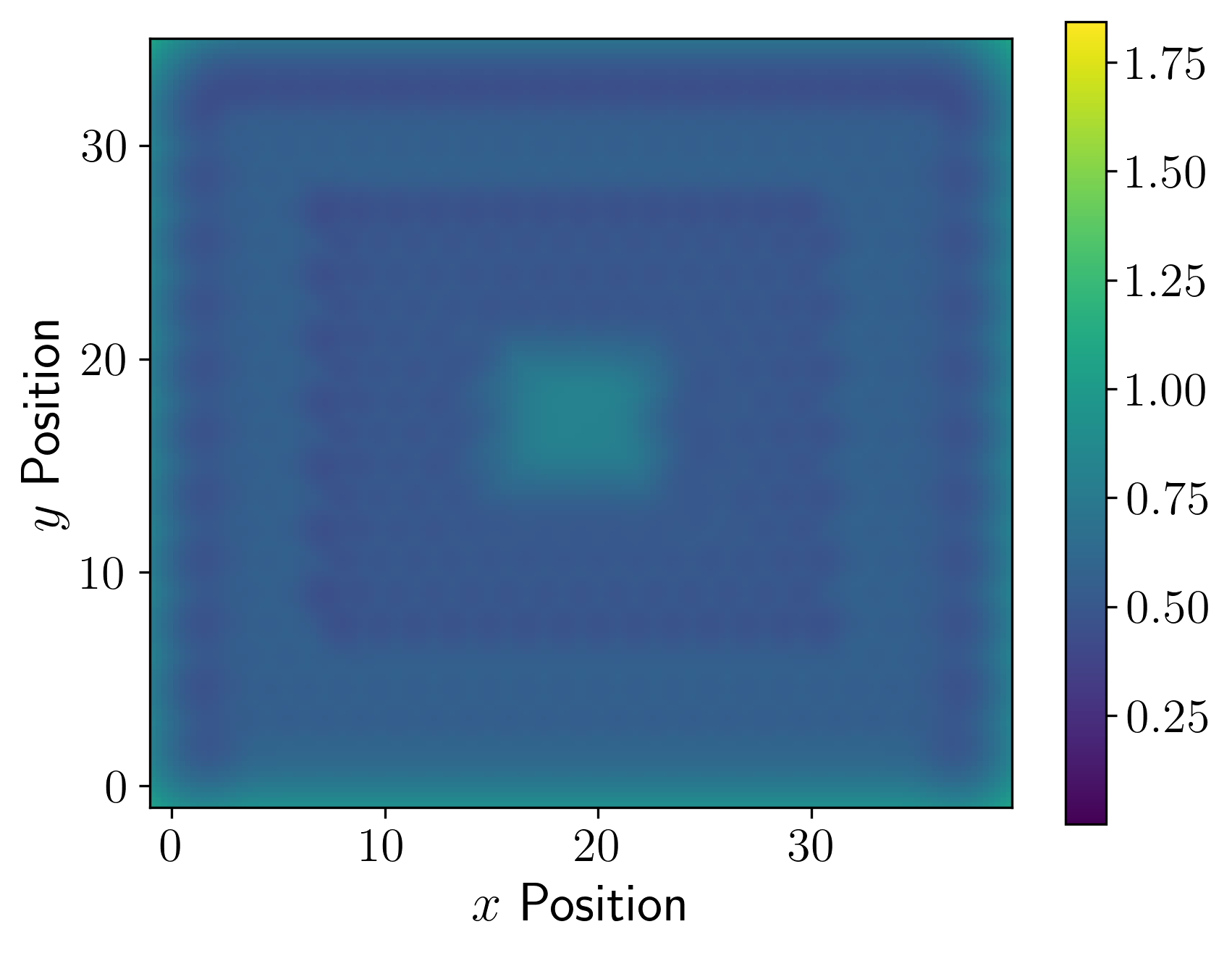}
    \end{minipage} &  \begin{minipage}{0.3\textwidth}
    \captionof{figure}{\\$\left\vert\bar{\mu}_{(\kappa\bm{x},1)}^{\text{C}}\left(\kappa\bm{X},H\right) - \mu_{(\bm{x},1)}^{\text{RQ}}\left(\kappa\bm{X},H\right)\right\vert$}\label{fig:hhs_gap_diff_e_1}
    \includegraphics[width=\textwidth]{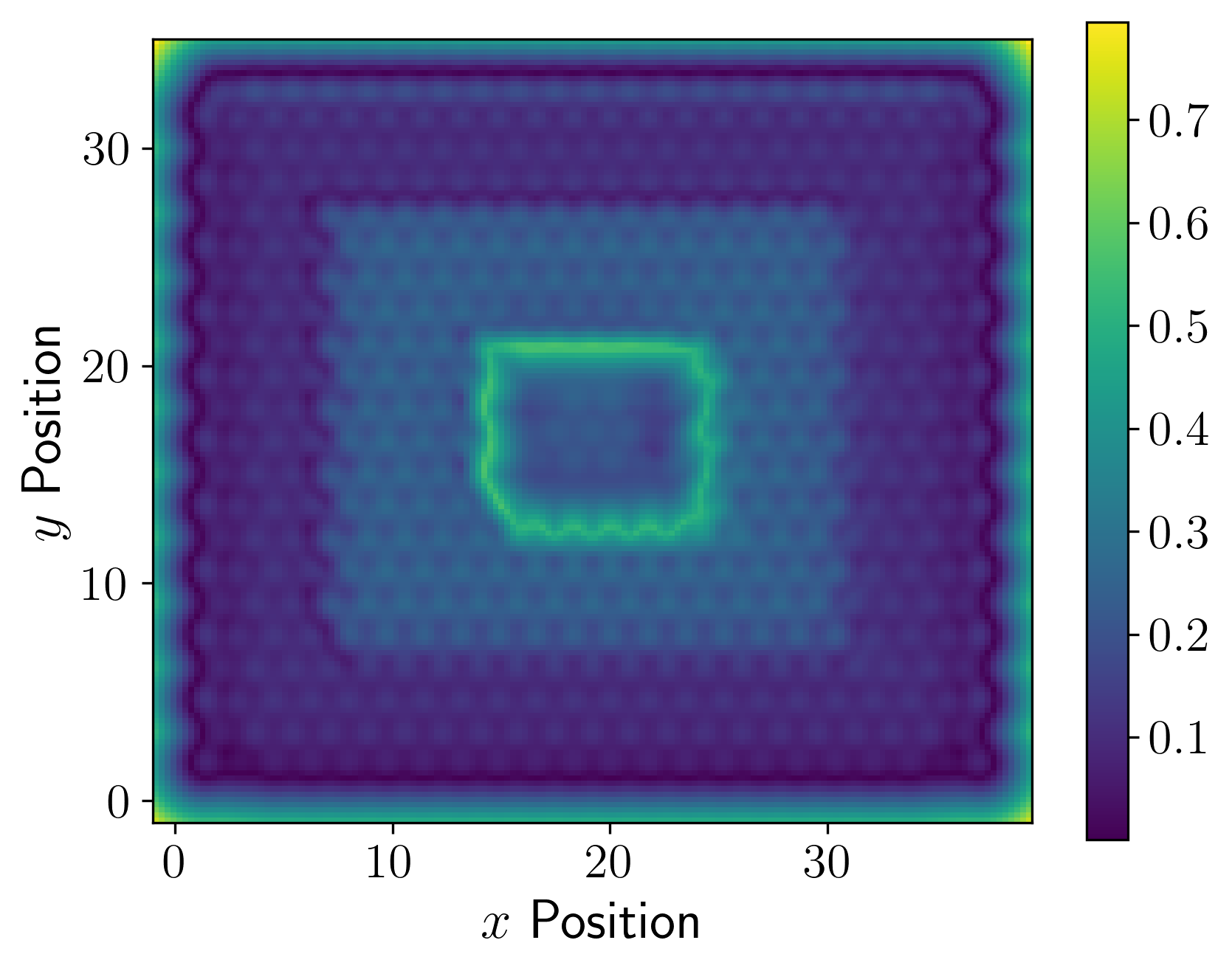}
    \end{minipage} \\
    \hline \hline

    \begin{minipage}{0.3\textwidth}
    \captionof{figure}{\\$\bar{\mu}_{(\kappa\bm{x},i)}^{\text{C}}\left(\kappa\bm{X},H\right)$}\label{fig:hhs_linear_gap_e_i}
    \includegraphics[width=\textwidth]{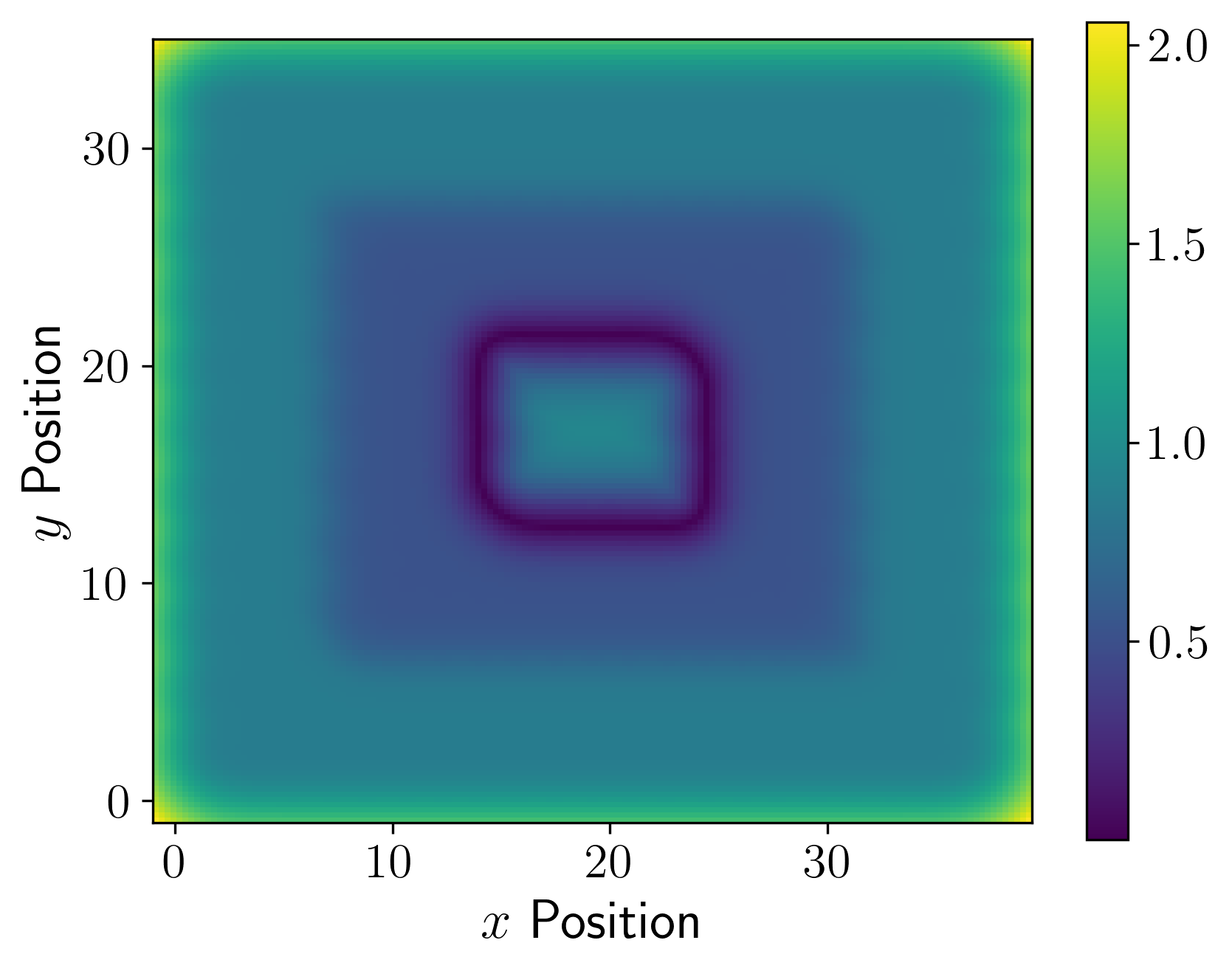}
    \end{minipage} &
    \begin{minipage}{0.3\textwidth}
    \captionof{figure}{\\$\mu_{(\kappa\bm{x},i)}^{\text{RQ}}\left(\kappa\bm{X},H\right)$}\label{fig:hhs_quadratic_gap_e_i}
    \includegraphics[width=\textwidth]{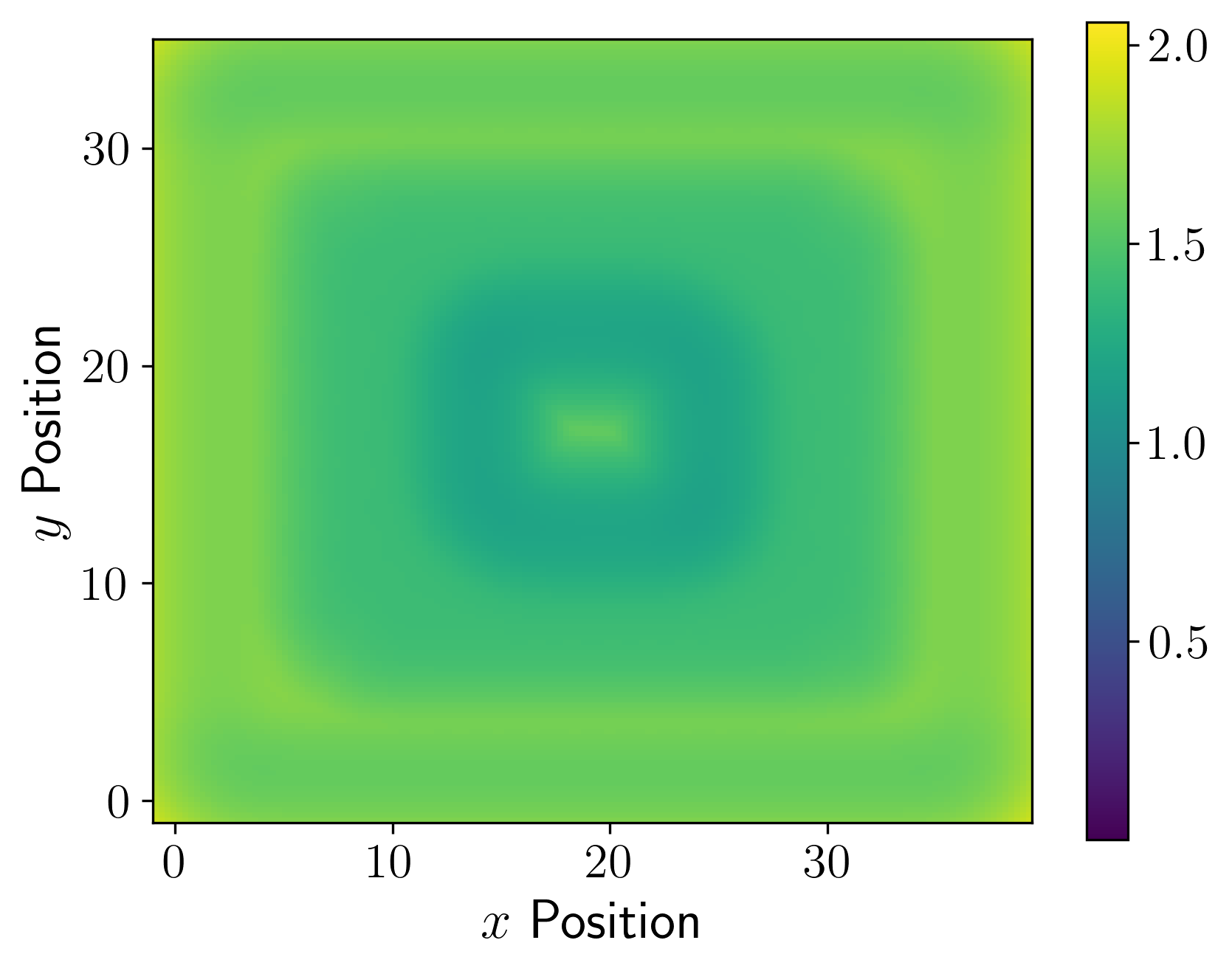}
    \end{minipage} &  \begin{minipage}{0.3\textwidth}
    \captionof{figure}{\\$\left\vert\bar{\mu}_{(\kappa\bm{x},i)}^{\text{C}}\left(\kappa\bm{X},H\right) - \mu_{(\bm{x},i)}^{\text{RQ}}\left(\kappa\bm{X},H\right)\right\vert$}\label{fig:hhs_gap_diff_e_i}
    \includegraphics[width=\textwidth]{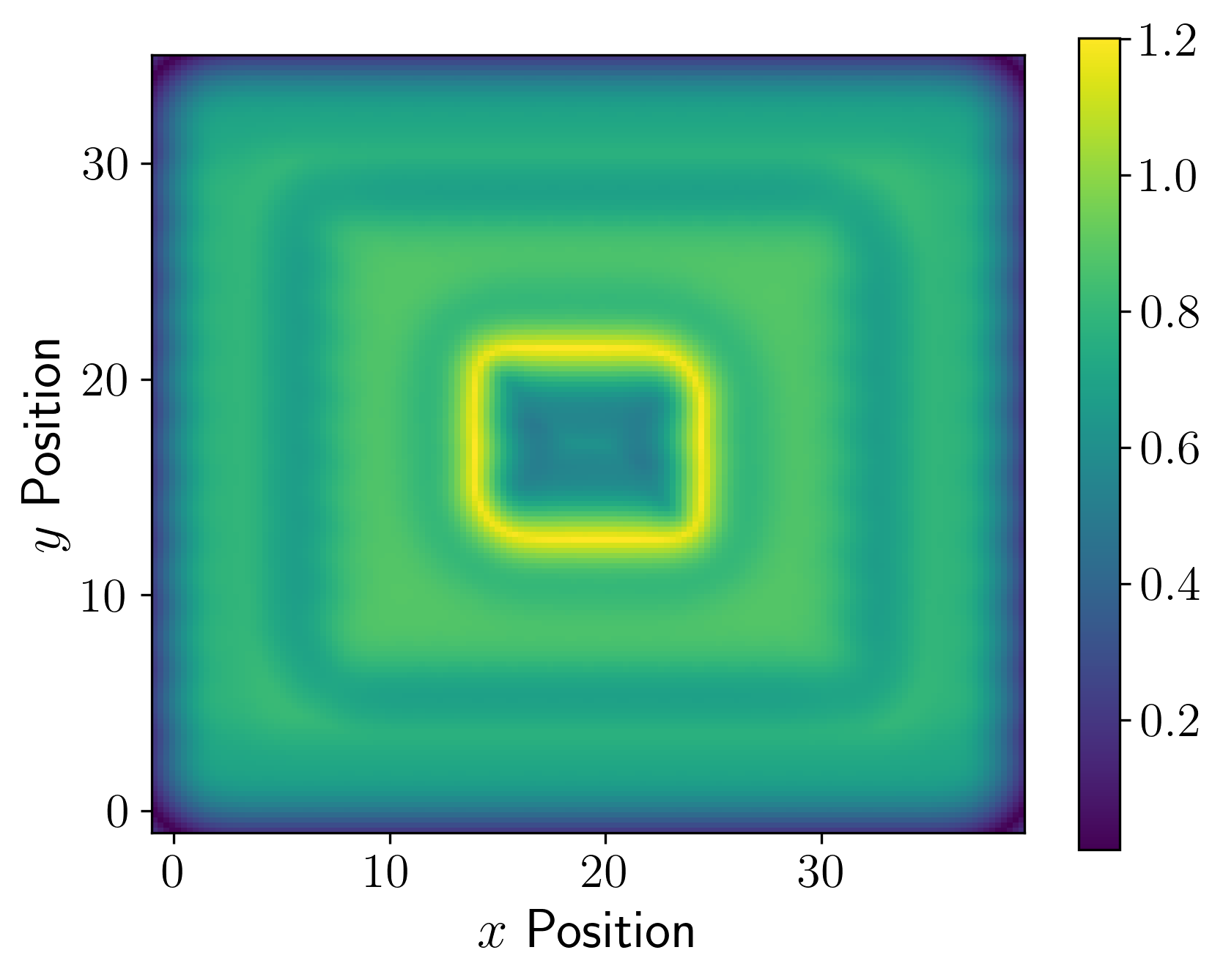}
    \end{minipage} \\ \hline \hline
    \end{tabular}
    \caption{Shown are the Clifford linear gap (Figure \ref{fig:hhs_linear_gap},\ref{fig:hhs_linear_gap_e_1}, and \ref{fig:hhs_linear_gap_e_i}), right quadratic gap (Figure \ref{fig:hhs_quadratic_gap},\ref{fig:hhs_quadratic_gap_e_1}, and \ref{fig:hhs_quadratic_gap_e_i}) as well as the difference between the gaps in Figure \ref{fig:hhs_gap_diff},\ref{fig:hhs_gap_diff_e_1},and \ref{fig:hhs_gap_diff_e_i} of the Haldane Heterostructure with $(\bm{x},E)=(x,y,E)$ being the probe site with varying energy prove site $E=0,1,i$ along the rows of the table. Note that $(\bm{X},H)=(X,Y,H)$ are the two Position operators and Hamiltonian.} \label{tbl:hhs_gaps_and_diff}
\end{table}

\section{Conclusion}
We have extended the theory of $\epsilon$-pseudospectrum \cite{cerjan_loring_vides_2023} to allow for the study of systems that have more than one non-Hermitian operator.
Along the way we have also proven some propositions about a special case of the non-Hermitian spectral localizer in \cite{cerjan_koekebier_Schulz_Baldes} which we call the non-Hermitian spectral localizer.
We have also shown that even with perturbations to non-Hermitian $B_i$ the quadratic gap is local.
With that said the benefits of the quadratic composite operator is that minimizing the corresponding gap function allows us to minimize eigen errors for various observables/matrices whithout increasing the computational complexity of the problem.

\section{Acknowledgements}
J.J.G.\ acknowledges support in part by an appointment to the Minority Educational Institution Student Partnership Program
(MEISPP), sponsored by the U.S. Department of Energy and administered by the Oak Ridge Institute for Science and Education, and the Center for Integrated Nanotechnologies Summer Research Initiative.
T.L.\ acknowledges support from the National Science Foundation, grant DMS-2349959. 
A.C.\ acknowledges support from the Laboratory Directed Research and Development program at Sandia National Laboratories.
This work was performed, in part, at the Center for Integrated Nanotechnologies, an Office of Science User Facility operated for the U.S. Department of Energy (DOE) Office of Science. Sandia National Laboratories is a multimission laboratory managed and operated by National Technology \& Engineering Solutions of Sandia, LLC, a wholly owned subsidiary of Honeywell International, Inc., for the U.S.\ DOE's National Nuclear Security Administration under contract DE-NA-0003525. The views expressed in the article do not necessarily represent the views of the U.S.\ DOE or the United States Government.

\printbibliography

@article{JHANJEE2017,
author = "SANGEETA JHANJEE",
title = "{Joint Spectral Theory Using Clifford Algebra}",
year = "2017",
month = "9",
url = "https://bridges.monash.edu/articles/thesis/Joint_Spectral_Theory_Using_Clifford_Algebra/5373127",
doi = "10.4225/03/59ae10e2b5ef5"
}

@article{kiorpelidis2024scaling,
title={Scaling of pseudospectra in exponentially sensitive lattices},
author={Kiorpelidis, Ioannis and Makris, Konstantinos G},
journal={arXiv preprint arXiv:2409.12036},
year={2024},
Author = {Kiorpelidis, Ioannis and Makris, Konstantinos G.},
Title = {Scaling of pseudospectra in exponentially sensitive lattices.},
URL = {https://libproxy.unm.edu/login?url=https://search.ebscohost.com/login.aspx?direct=true&amp;db=edsarx&amp;AN=edsarx.2409.12036&amp;site=eds-live&amp;scope=site},
Year = {2024},
}

@article{hasan_kane_2010_RevModPhys.82.3045,
  title = {Colloquium: Topological insulators},
  author = {Hasan, M. Z. and Kane, C. L.},
  journal = {Rev. Mod. Phys.},
  volume = {82},
  issue = {4},
  pages = {3045--3067},
  numpages = {0},
  year = {2010},
  month = {11},
  publisher = {American Physical Society},
  doi = {10.1103/RevModPhys.82.3045},
  url = {https://link.aps.org/doi/10.1103/RevModPhys.82.3045}
}

@article{xiao_chang_niu_RevModPhys.82.1959,
  title = {Berry phase effects on electronic properties},
  author = {Xiao, Di and Chang, Ming-Che and Niu, Qian},
  journal = {Rev. Mod. Phys.},
  volume = {82},
  issue = {3},
  pages = {1959--2007},
  numpages = {0},
  year = {2010},
  month = {07},
  publisher = {American Physical Society},
  doi = {10.1103/RevModPhys.82.1959},
  url = {https://link.aps.org/doi/10.1103/RevModPhys.82.1959}
}

@article{bansil_lin_das_RevModPhys.88.021004,
  title = {Colloquium: Topological band theory},
  author = {Bansil, A. and Lin, Hsin and Das, Tanmoy},
  journal = {Rev. Mod. Phys.},
  volume = {88},
  issue = {2},
  pages = {021004},
  numpages = {37},
  year = {2016},
  month = {06},
  publisher = {American Physical Society},
  doi = {10.1103/RevModPhys.88.021004},
  url = {https://link.aps.org/doi/10.1103/RevModPhys.88.021004}
}

@article{ozawa_et_all_RevModPhys.91.015006,
  title = {Topological photonics},
  author = {Ozawa, Tomoki and Price, Hannah M. and Amo, Alberto and Goldman, Nathan and Hafezi, Mohammad and Lu, Ling and Rechtsman, Mikael C. and Schuster, David and Simon, Jonathan and Zilberberg, Oded and Carusotto, Iacopo},
  journal = {Rev. Mod. Phys.},
  volume = {91},
  issue = {1},
  pages = {015006},
  numpages = {76},
  year = {2019},
  month = {03},
  publisher = {American Physical Society},
  doi = {10.1103/RevModPhys.91.015006},
  url = {https://link.aps.org/doi/10.1103/RevModPhys.91.015006}
}

@article{trefethen_reddy_driscoll_1993_6dce06aa-3a1d-396a-ad07-b01838717a11,
 ISSN = {00368075, 10959203},
 URL = {http://www.jstor.org/stable/2882016},
 author = {Lloyd N. Trefethen and Anne E. Trefethen and Satish C. Reddy and Tobin A. Driscoll},
 journal = {Science},
 number = {5121},
 pages = {578--584},
 publisher = {American Association for the Advancement of Science},
 title = {Hydrodynamic Stability Without Eigenvalues},
 urldate = {2024-07-10},
 volume = {261},
 year = {1993}
}

@article{trefethen_1997_edsjsr.213303719970901,
Author = {Trefethen, Lloyd N.},
ISSN = {00361445},
Journal = {SIAM Review},
Number = {3},
Pages = {383 - 406},
Title = {Pseudospectra of Linear Operators.},
Volume = {39},
URL = {https://libproxy.unm.edu/login?url=https://search.ebscohost.com/login.aspx?direct=true&amp;db=edsjsr&amp;AN=edsjsr.2133037&amp;site=eds-live&amp;scope=site},
Year = {1997},
}

@book{trefethen_embree_2005_MR2155029,
    AUTHOR = {Trefethen, Lloyd N. and Embree, Mark},
     TITLE = {Spectra and pseudospectra},
      NOTE = {The behavior of nonnormal matrices and operators},
 PUBLISHER = {Princeton University Press, Princeton, NJ},
      YEAR = {2005},
     PAGES = {xviii+606},
      ISBN = {978-0-691-11946-5},
   MRCLASS = {15-02 (15A18 47A10 47A50)},
  MRNUMBER = {2155029},
MRREVIEWER = {David\ Scott\ Watkins},
}

@article{krejcirik_siegl_tater_viola_2015_10.1063/1.4934378,
    author = {Krejčiřík, D. and Siegl, P. and Tater, M. and Viola, J.},
    title = "{Pseudospectra in non-Hermitian quantum mechanics}",
    journal = {Journal of Mathematical Physics},
    volume = {56},
    number = {10},
    pages = {103513},
    year = {2015},
    month = {10},
    issn = {0022-2488},
    doi = {10.1063/1.4934378},
    url = {https://doi.org/10.1063/1.4934378},
    eprint = {https://pubs.aip.org/aip/jmp/article-pdf/doi/10.1063/1.4934378/13907006/103513\_1\_online.pdf},
}

@article{loring_schulz-baldes_2017_12504413820170101,
Author = {Loring, Terry A. and Schulz-Baldes, Hermann},
ISSN = {10769803},
Journal = {New York Journal of Mathematics},
Pages = {1111 - 1140},
Title = {Finite volume calculation of K-theory invariants.},
Volume = {23},
URL = {https://libproxy.unm.edu/login?url=https://search.ebscohost.com/login.aspx?direct=true&amp;db=a9h&amp;AN=125044138&amp;site=eds-live&amp;scope=site},
Year = {2017},
}

@article{loring_schulz-baldes_2020_edsgcl.62678504020200301,
Author = {Loring, Terry A. and Schulz-Baldes, Hermann},
ISSN = {1661-6952},
Journal = {Journal of Noncommutative Geometry},
Number = {1},
Pages = {1},
Title = {The spectral localizer for even index pairings.},
Volume = {14},
URL = {https://libproxy.unm.edu/login?url=https://search.ebscohost.com/login.aspx?direct=true&amp;db=edsgao&amp;AN=edsgcl.626785040&amp;site=eds-live&amp;scope=site},
Year = {2020},
}

@article{cerjan_loring_2024_127892,
title = {Even spheres as joint spectra of matrix models},
journal = {Journal of Mathematical Analysis and Applications},
volume = {531},
number = {1, Part 2},
pages = {127892},
year = {2024},
issn = {0022-247X},
doi = {https://doi.org/10.1016/j.jmaa.2023.127892},
url = {https://www.sciencedirect.com/science/article/pii/S0022247X23008958},
author = {Alexander Cerjan and Terry A. Loring}
}

@article{cerjan_prl_2024,
  title={Local markers for crystalline topology},
  author={Cerjan, Alexander and Loring, Terry A and Schulz-Baldes, Hermann},
  journal={Physical Review Letters},
  volume={132},
  number={7},
  pages={073803},
  year={2024},
  publisher={APS}
}

@article{kawabata_shiozaki_ueda_sato_PhysRevX.9.041015_2019,
  title = {Symmetry and Topology in Non-Hermitian Physics},
  author = {Kawabata, Kohei and Shiozaki, Ken and Ueda, Masahito and Sato, Masatoshi},
  journal = {Phys. Rev. X},
  volume = {9},
  issue = {4},
  pages = {041015},
  numpages = {52},
  year = {2019},
  month = {10},
  publisher = {American Physical Society},
  doi = {10.1103/PhysRevX.9.041015},
  url = {https://link.aps.org/doi/10.1103/PhysRevX.9.041015}
}

@article{cerjan_loring_vides_2023,
  author =       {Cerjan, Alexander and Loring, Terry and Vides, Fredy},
  title =        {Quadratic pseudospectrum for identifying localized states},
  journal =      {Journal of Mathematical Physics},
  volume =       {64},
  number =       {2},
  pages =        {023501},
  YEAR =         {2023},
  DOI =          {https://doi.org/10.1063/5.0098336}
}

@article{kahlil_cerjan_loring,
  title = {Classifying Topology in Photonic Heterostructures with Gapless Environments},
  author = {Dixon, Kahlil Y. and Loring, Terry A. and Cerjan, Alexander},
  journal = {Phys. Rev. Lett.},
  volume = {131},
  issue = {21},
  pages = {213801},
  numpages = {7},
  year = {2023},
  month = {11},
  publisher = {American Physical Society},
  doi = {10.1103/PhysRevLett.131.213801},
  url = {https://link.aps.org/doi/10.1103/PhysRevLett.131.213801}
}

@article{cerjan_loring_2022,
  title = {Local invariants identify topology in metals and gapless systems},
  author = {Cerjan, Alexander and Loring, Terry A.},
  journal = {Phys. Rev. B},
  pages = {064109},
  numpages = {10},
  year = {2022},
  month = {08},
  publisher = {American Physical Society},
  doi = {10.1103/PhysRevB.106.064109},
  url = {https://link.aps.org/doi/10.1103/PhysRevB.106.064109}
}

@article{cerjan_loring_pho_2022,
url = {https://doi.org/10.1515/nanoph-2022-0547},
title = {An operator-based approach to topological photonics},
author = {Alexander Cerjan and Terry A. Loring},
pages = {4765--4780},
volume = {11},
number = {21},
journal = {Nanophotonics},
doi = {doi:10.1515/nanoph-2022-0547},
year = {2022},
lastchecked = {2024-07-08}
}

@article{loring,
  author =       {Terry, Loring},
  title =        {K-theory and pseudospectra for topological insulators},
  journal =      {Annals of Physics},
  volume =       {356},
  pages =        {383-416},
  year =         {2015},
  DOI =          {https://doi.org/10.1016/j.aop.2015.02.031}
}

@article{cerjan_koekebier_Schulz_Baldes,
  author =       {Cerjan, Alexander and Koekenbier, Lars and Schulz-Baldes, Hermann},
  title =        {Spectral localizer for line-gapped non-hermitian systems},
  journal =      {Journal of Mathematical Physics},
  volume =       {64},
  number =       {8},
  pages =        {082102},
  year =         {2023},
  DOI =          {https://doi.org/10.1063/5.0150995}
}

@article{ruhe,
  author =       {Ruhe, Axel},
  title =        {On the Closeness of Elgenvalues and Singular Values for Almost Normal Matrices},
  journal =      {Linear Algebra and it's Applications},
  volume =       {11},
  pages =        {87--93},
  year =         {1975},
  DOI =          {https://doi.org/10.1016/0024-3795(75)90119-6.}
}

@article{elsner_paardekooper,
  author =       {Elsner, L. and Paardekooper, M.H.C.},
  title =        {On Measures of Nonnormality of Matrices},
  journal =      {Linear Algebra and it's Applications},
  volume =       {92},
  pages =        {107--124},
  year =         {1987},
  DOI =          {https://doi.org/10.1016/0024-3795(87)90253-9}
}

@article{kahan,
  author =       {W, Kahan},
  title =        {Spectra of Nearly Hermitian Matrices},
  journal =      {Proceedings of the American Mathematical Society},
  volume =       {48},
  number =       {1},
  pages =        {11--17},
  year =         {1975},
  DOI =          {https://doi.org/10.2307/2040683}
}

@article{haldane,
  author =       {Haldane, Frederick D.},
  title =        {Model for a Quantum Hall Effect without Landau Levels: Condensed-Matter Realization of the ``Parity Anomaly''},
  journal =      {Physical Review Letters},
  volume =       {61},
  number =       {18},
  pages =        {2015--2018},
  year =         {1988},
  DOI =          {https://doi-org.libproxy.unm.edu/10.1103/PhysRevLett.61.2015}
}

@article{tls_paper,
  author =       {Ramy, El-Ganainy and Konstantinos, G. Makris and Mercedeh, Khajavikhan and Ziad, H. Musslimani and Stefan, Rotter and Demetrios, N. Christodoulides},
  title =        {Non-Hermitian physics and PT symmetry},
  journal =      {Nature Physics},
  volume =       {14},
  pages =        {11-19},
  year =         {2018},
  DOI =          {https://doi.org/10.1038/nphys4323}
}

@book{hirvensalo_2001,
Author = {Hirvensalo, Mika},
ISBN = {3540667830},
Publisher = {Springer},
Series = {Natural computing series},
Title = {Quantum computing.},
URL = {https://libproxy.unm.edu/login?url=https://search.ebscohost.com/login.aspx?direct=true&amp;db=cat05987a&amp;AN=unm.46538501&amp;site=eds-live&amp;scope=site},
Year = {2001},
}

@book{golub_vanloan,
  title={Matrix Computations},
  author={Golub, Gene H. and Van Loan, Charles F.},
  isbn={9781421407944},
  series={Johns Hopkins Studies in the Mathematical Sciences, volume 3},
  year={2013},
  publisher={Johns Hopkins University Press}
}

@book{bhatia,
  title={Matrix Analysis},
  author={Rajendra, Bhatia},
  isbn={9780387948461},
  series={Graduate Texts in Mathematics 169},
  year={1997},
  publisher={Springer-Verlag New York}
}

@book{horn_johnson,
  title={Matrix Analysis},
  author={Roger A., Horn and Charles R. Johnson},
  isbn={9780521305860},
  year={1985},
  publisher={Cambridge University Press}
}

@article{Esaki_Sato_Hasebe_Kohmoto_PR_2011,
  title = {Edge states and topological phases in non-Hermitian systems},
  author = {Esaki, Kenta and Sato, Masatoshi and Hasebe, Kazuki and Kohmoto, Mahito},
  journal = {Phys. Rev. B},
  volume = {84},
  issue = {20},
  pages = {205128},
  numpages = {19},
  year = {2011},
  month = {11},
  publisher = {American Physical Society},
  doi = {10.1103/PhysRevB.84.205128},
  url = {https://link.aps.org/doi/10.1103/PhysRevB.84.205128}
}

@article{Shen_Zhen_Fu_PRL_2018,
  title = {Topological Band Theory for Non-Hermitian Hamiltonians},
  author = {Shen, Huitao and Zhen, Bo and Fu, Liang},
  journal = {Phys. Rev. Lett.},
  volume = {120},
  issue = {14},
  pages = {146402},
  numpages = {6},
  year = {2018},
  month = {04},
  publisher = {American Physical Society},
  doi = {10.1103/PhysRevLett.120.146402},
  url = {https://link.aps.org/doi/10.1103/PhysRevLett.120.146402}
}

@article{Kawabata_etal_nature_2019,
  title={Topological unification of time-reversal and particle-hole symmetries in non-Hermitian physics},
  author={Kawabata, Kohei and Higashikawa, Sho and Gong, Zongping and Ashida, Yuto and Ueda, Masahito},
  journal={Nature communications},
  volume={10},
  number={1},
  pages={297},
  year={2019},
  publisher={Nature Publishing Group UK London}
}

@article{purcell_spontaneous_1946,
	title = {Spontaneous emission probabilities at radio frequencies},
	volume = {69},
	journal = {Phys. Rev.},
	author = {Purcell, E. M.},
	year = {1946},
	pages = {681},
}

@article{Arnold1972,
  author = {Arnol\textquotesingle d, V. I.},
  title = {Modes and Quasimodes},
  journal = {Functional Analysis and Its Applications},
  volume = {6},
  number = {2},
  pages = {94-101},
  year = {1972},
  url = {https://doi.org/10.1007/BF01077511}
}

@article{landau_1976,
  author = {Landau, H. J.},
  title = {{Loss in unstable resonators.}},
  journal = {Journal of the Optical Society of America},
  volume = {66},
  pages = {529-525},
  year = {1976},
  url = {https://research.ebsco.com/linkprocessor/plink?id=4a1855b7-6ccd-3d4b-9df0-d12789ef43da}
}

@article{landau_1977,
  author = {Landau, Henry},
  publisher = {American Mathematical Society (AMS)},
  title = {The Notion of Approximate Eigenvalues Applied to an Integral
    Equation of Laser Theory},
  journal = {Quarterly of Applied Mathematics},
  volume = {35},
  pages = {172-165},
  year = {1977},
  url = {https://research.ebsco.com/linkprocessor/plink?id=2fa66e2d-e76a-36ff-af48-6cc007bdc8b4}
}

@article{Trefethen1992,
  author    = {N. Trefethen},
  title     = {Pseudospectra of matrices},
  journal   = {Numerical Analysis 1991},
  editor    = {D. F. Griffiths and G. A. Watson},
  pages     = {234--266},
  publisher = {Longman Scientific and Technical},
  address   = {Harlow, Essex, UK},
  year      = {1992},
  note      = {(46, 58, 71, 236, 371, 381, 451, 452)}
}

\end{document}